\documentclass[11pt]{article}

\usepackage{wrapfig}
\usepackage{graphicx}
\usepackage{palatino}
\usepackage{amsmath}               
\usepackage{amsfonts}              
\usepackage{amsthm}                
\usepackage{amssymb}
\usepackage{todonotes}
\usepackage[pagebackref]{hyperref}
\usepackage[margin=1in]{geometry}
\usepackage{todonotes}

\newtheorem{thm}{Theorem}[section]
\newtheorem{lem}[thm]{Lemma}

\newtheorem{cor}[thm]{Corollary}
\newtheorem{conj}[thm]{Conjecture}
\newtheorem{defi}[thm]{Definition}

\DeclareMathOperator*{\E}{\mathbb{E}}
\DeclareMathOperator*{\Cov}{\mathsf{Cov}}
\DeclareMathOperator*{\Var}{\mathsf{Var}}
\newcommand{\Inf}{\mathsf{Inf}}
\newcommand{\w}{\mathsf{wt}}
\newcommand{\len}{\mathsf{len}}
\newcommand{\dist}{\mathsf{dist}}

\newcommand{\RR}{\mathbb{R}}      
\newcommand{\NN}{\mathbb{N}}      
\newcommand{\ZZ}{\mathbb{Z}}      

\newcommand{\calD}{\mathcal{D}}
\newcommand{\VD}{V_\mathcal{D}}
\newcommand{\ED}{E_\mathcal{D}}
\newcommand{\bari}{\overline{i}}
\newcommand{\bara}{\overline{\alpha}}
\newcommand{\barc}{\overline{c}}
\newcommand{\SE}{\mathsf{E}}
\newcommand{\SV}{\mathsf{V}}
\newcommand{\SL}{\mathsf{L}}
\newcommand{\SM}{\mathsf{M}}
\newcommand{\SF}{\mathsf{F}}
\newcommand{\SB}{\mathsf{B}}
\newcommand{\SSS}{\mathsf{S}}
\newcommand{\SLF}{\mathsf{LF}}
\newcommand{\SSF}{\mathsf{SF}}
\newcommand{\SLB}{\mathsf{LB}}
\newcommand{\SSB}{\mathsf{SB}}
\newcommand{\OPT}{\mathsf{OPT}}

\newcommand{\edgecut}{Short Path Edge Cut }
\newcommand{\vertexcut}{Short Path Vertex Cut }
\newcommand{\edgecutn}{Short Path Edge Cut}
\newcommand{\vertexcutn}{Short Path Vertex Cut}

\begin{document}

\setcounter{page}{0}

\title{{\bf Improved Hardness for Cut, Interdiction, \\ 
and Firefighter Problems}}

\author{
Euiwoong Lee\thanks{Supported by the Samsung Scholarship, the Simons Award for Graduate Students in TCS, and Venkat Guruswami's NSF CCF-1115525. {\tt euiwoonl@cs.cmu.edu} }}

\date{Computer Science Department \\ Carnegie Mellon University \\ Pittsburgh, PA 15213.}

\maketitle

\begin{abstract}
We study variants of the classic $s$-$t$ cut problem and prove the following improved hardness results assuming the Unique Games Conjecture (UGC). 
\begin{itemize}
\item For any constant $k \geq 2$ and $\epsilon > 0$, we show that {\em Directed Multicut} with $k$ source-sink pairs is hard to approximate within a factor $k - \epsilon$. This matches the trivial $k$-approximation algorithm. By a simple reduction, our result for $k = 2$ implies that {\em Directed Multiway Cut} with two terminals (also known as {\em $s$-$t$ Bicut}) is hard to approximate within a factor $2 - \epsilon$, matching the trivial $2$-approximation algorithm. 
Previously, the best hardness factor for these problems (for constant $k$) was $1.5 - \epsilon$~\cite{EVW13, CM16} under the UGC. 


\item For {\em Length-Bounded Cut} and {\em Shortest Path Interdiction}, we show that both problems are hard to approximate within any constant factor, even if we allow bicriteria approximation. If we want to cut vertices or the graph is directed, our hardness factor for Length-Bounded Cut matches the best approximation ratio up to a constant. 
Previously, the best hardness factor was $1.1377$ for Length-Bounded Cut~\cite{BEHKKPSS10} and $2$ for Shortest Path Interdiction~\cite{KBBEGRZ07}.

\item Assuming a variant of the UGC (implied by another variant of Bansal and Khot~\cite{BK09}), we prove that it is hard to approximate {\em Resource Minimization Fire Containment} within any constant factor. Previously, the best hardness factor was $2$~\cite{KM10}. 
\end{itemize}

Our results are based on a general method of {\em converting an integrality gap instance to a length-control dictatorship test} for variants of the $s$-$t$ cut problem, which may be useful for other problems.
\end{abstract}

\thispagestyle{empty}

\newpage

\section {Introduction}
One of the most important implications of the Unique Games Conjecture (UGC,~\cite{Khot02}) is the results of Khot et al.~\cite{KKMO07} and Raghavendra~\cite{Raghavendra08}, which say that for any maximum constraint satisfaction problem (Max-CSP), an integrality gap instance of the standard semidefinite programming (SDP) relaxation can be converted to the NP-hardness result with the same gap. These results initiated the study of beautiful connections between power of convex relaxations and hardness of approximation, from which surprising results for both subjects have been discovered. 

While their results hold for problems in Max-CSPs, the framework of {\em converting an integrality gap instance to hardness} has been successfully applied to covering and graph cut problems.
For graph cut problems, Manokaran et al.~\cite{MNRS08} showed that for {\em Undirected Multiway Cut} and its generalizations, an integrality gap of the standard linear programming (LP) relaxation implies the hardness result assuming the UGC. Their result is further generalized by Ene et al.~\cite{EVW13} by formulating them as {\em Min-CSPs}. 
On the other hand, Kumar et al.~\cite{KMTV11} studied {\em Strict CSPs} and showed the same phenomenon for the standard LP relaxation. 

One of the limitations of the previous CSP-based transformations from LP gap instances to hard instances is based on the fact that they do not usually preserve the desired structure of the constraint hypergraph.\footnote{One of notable exceptions we are aware is the result of Guruswami et al.~\cite{GSS15}, using Kumar et al.~\cite{KMTV11} to show that $k$-Uniform $k$-Partite Hypergraph Vertex Cover is hard to approximate within a factor $\frac{k}{2} - \epsilon$ for any $\epsilon > 0$.} For example, consider the {\em Length-Bounded Edge Cut} problem where the input consists of a graph $G = (V, E)$, two vertices $s, t \in V$, and a constant $l \in \NN$, and the goal is to remove the fewest edges to ensure there is no path from $s$ to $t$ of length less than $l$. 
This problem can be viewed as a special case of {\em Hypergraph Vertex Cover} (HVC)
by viewing each edge as a vertex of a hypergraph and creating a hyperedge for every $s$-$t$ path of length less than $l$. 
While HVC is in turn a Strict CSP, but its integrality gap instance cannot be converted to hardness using Kumar et al.~\cite{KMTV11} as a black-box, since the set of hyperedges created in the resulting hard instance is not guaranteed to correspond to the set of short $s$-$t$ paths of some graph. 

For Undirected Multiway Cut, Manokaran et al.~\cite{MNRS08} bypassed this difficulty by using $2$-ary constraints so that the resulting constraint hypergraph becomes a graph again. 
For {\em Undirected Node-weighted Multiway Cut}, Ene et al.~\cite{EVW13} used the equivalence to {\em Hypergraph Multiway Cut}~\cite{OFN12} so that the resulting hypergraph does not need to satisfy additional structure. 
These problems are then formulated as a Min-CSP by using many labels which are supposed to represent different connected components. For {\em Directed Multiway Cut}, to the best of our knowledge, the existence of an analogous formulation as a Min-CSP is unknown. 

We study variants of the classical $s$-$t$ cut problem in both directed and undirected graphs that have been actively studied, including the aforementioned Length-Bounded Cut and Directed Multiway Cut. 
We prove the optimal hardness or the first super-constant hardness for them. 
See Section~\ref{subsec:results} for the definitions of the problems and our results. 
All our results are based on the general framework of converting an integrality gap instance to a {\em length-control dictatorship test}. 
The structure of our length-control dictatorship tests allows us to naturally convert an integrality gap instance for the basic LP for various cut problems to hardness based on the UGC. 
Section~\ref{subsec:techniques} provides more detailed intuition of this framework. 
While these problems have slightly different characteristics that make it hard to present the single result for a wide class of problems like CSPs, we hope that our techniques may be useful to prove hardness of other cut problems.

\subsection{Problems and Results}
\label{subsec:results}

\paragraph{Directed Multicut and Directed Multiway Cut.}
Given a directed graph and two vertices $s$ and $t$, one of the most natural variants of $s$-$t$ 
cut is to remove the fewest edges to ensure that there is no directed path from $s$ to $t$ and no directed path from $t$ to $s$. This problem is known as {\em $s$-$t$ Bicut} and admits the trivial $2$-approximation algorithm by computing the minimum $s$-$t$ cut and $t$-$s$ cut. 

{\em Directed Multiway Cut} is a generalization of $s$-$t$ Bicut that has been actively studied. Given a directed graph with $k$ terminals $s_1, \dots, s_k$, the goal is to remove the fewest number of edges such that there is no path from $s_i$ to $s_j$ for any $i \neq j$. Directed Multiway Cut also admits $2$-approximation~\cite{NZ01, CM16}. 
If $k$ is allowed to increase polynomially with $n$, there is a simple reduction from Vertex Cover that shows $(2 - \epsilon)$-approximation is hard under the UGC~\cite{GVY94, KR08}. 

Directed Multiway Cut can be further generalized to {\em Directed Multicut}. Given a directed graph with $k$ source-sink pairs $(s_1, t_1), \dots, (s_k, t_k)$, the goal is to remove the fewest number of edges such that there is no path from $s_i$ to $t_i$ for any $i$. 
Computing the minimum $s_i$-$t_i$ cut for all $i$ separately gives the trivial $k$-approximation algorithm. 
Chuzhoy and Khanna~\cite{CK09} showed Directed Multicut is hard to approximate within a factor $2^{\Omega(\log^{1 - \epsilon} n)} = 2^{\Omega(\log^{1 - \epsilon} k)}$ when $k$ is polynomially growing with $n$. Agarwal et al.~\cite{AAC07} showed $\tilde{O}(n^{\frac{11}{23}})$-approximation algorithm, which improves the trivial $k$-approximation when $k$ is large.

Very recently, Chekuri and Madan~\cite{CM16} showed simple approximation-preserving reductions from Directed Multicut with $k = 2$ to $s$-$t$ Bicut (the other direction is trivially true), and {\em (Undirected) Node-weighted Multiway Cut} with $k = 4$ to $s$-$t$ Bicut. 
Since Node-weighted Multiway Cut with $k = 4$ is hard to approximate within a factor $1.5 - \epsilon$ under the UGC~\cite{EVW13} (matching the algorithm of Garg et al.~\cite{GVY94}), the same hardness holds for $s$-$t$ Bicut, Directed Multiway Cut, and Directed Multicut for constant $k$. To the best of our knowledge, $1.5 - \epsilon$ is the best hardness factor for constant $k$ even assuming the UGC. In the same paper, Chekuri and Madan~\cite{CM16} asked whether a factor $2 - \epsilon$ hardness holds for $s$-$t$ Bicut under the UGC. 

We prove that for any constant $k \geq 2$, the trivial $k$-approximation for Directed Multicut might be optimal. Our result for $k = 2$ gives the optimal hardness result for $s$-$t$ Bicut, answering the question of Chekuri and Madan. 

\begin{thm}
Assuming the Unique Games Conjecture, for every $k \geq 2$ and $\epsilon > 0$, Directed Multicut with $k$ source-sink pairs is NP-hard to approximate within a factor $k - \epsilon$. 
\label{thm:multicut}
\end{thm}

\begin{cor}
Assuming the Unique Games Conjecture, for any $\epsilon > 0$, 
$s$-$t$ Bicut is hard to approximate within a factor $2 - \epsilon$. 
\end{cor}

\paragraph{Length-Bounded Cut and Shortest Path Interdiction.}
Another natural variant of $s$-$t$ cut is the {\em Length-Bounded Cut} problem, where given an integer $l$, we only want to cut $s$-$t$ paths of length strictly less than $l$.\footnote{It is more conventional to cut $s$-$t$ paths of length {\em at most} $l$.  We use this slightly nonconventional way to be more consistent with Shortest Path Interdiction.} Its practical motivation is based on the fact that in most communication / transportation networks, short paths are preferred to be used to long paths~\cite{MM10}. 

Lov{\'a}sz et al.~\cite{LNP78} gave an exact algorithm for Length-Bounded Vertex Cut $(l \leq 5)$ in undirected graphs. Mahjoub and McCormick~\cite{MM10} proved that Length-Bounded Edge Cut admits an exact polynomial time algorithm for $l \leq 4$ in undirected graphs. 
Baier et al.~\cite{BEHKKPSS10} showed that both Length-Bounded Vertex Cut $(l > 5)$ and Length-Bounded Edge Cut $(l > 4)$ are NP-hard to approximate within a factor $1.1377$. They presented $O(\min(l, \frac{n}{l})) = O(\sqrt{n})$-approximation algorithm for Length-Bounded Vertex Cut and $O(\min(l, \frac{n^2}{l^2}, \sqrt{m})) = O(n^{2/3})$-approximation algorithm for Length-Bounded Edge Cut, with matching LP gaps. 
Length-Bounded Cut problems have been also actively studied in terms of their fixed parameter tractability~\cite{GT11, DK15, BNN15, FHNN15}. 

If we exchange the roles of the objective $k$ and the length bound $l$, the problem becomes {\em Shortest Path Interdiction}, where we want to maximize the length of the shortest $s$-$t$ path after removing at most $k$ vertices or edges. 
It is also one of the central problems in a broader class of {\em interdiction} problems, where an {\em attacker} tries to remove some edges or vertices to destroy a desirable property (e.g., short $s$-$t$ distance, large $s$-$t$ flow, cheap MST) of a network (see the survey of ~\cite{SPG13}). 
The study of Shortest Path Interdiction started in 1980's when the problem was called as the {\em $k$-most-vital-arcs} problem~\cite{CD82, MMG89, BGV89} and proved to be NP-hard~\cite{BGV89}. Khachiyan et al.~\cite{KBBEGRZ07} proved that it is NP-hard to approximate within a factor less than $2$. 
While many heuristic algorithms were proposed~\cite{IW02, BB08, Morton11} and hardness in planar graphs~\cite{PS13} was shown, whether the general version admits a constant factor approximation was still unknown. 





Given a graph $G = (V, E)$ and $s, t \in V$, let $\dist(G)$ be the length of the shortest $s$-$t$ path. For $V' \subseteq V$, let $G \setminus V'$ be the subgraph induced by $V \setminus V'$. For $E' \subseteq E$, we use the same notation $G \setminus E'$ to denote the subgraph $(V, E \setminus E')$. We primarily study undirected graphs. We first present our results for the vertex version of both problems (collectively called as {\em \vertexcut} onwards). 

\begin{thm}
Assume the Unique Games Conjecture. 
For infinitely many values of $l \in \NN$, 
given an undirected graph $G = (V, E)$ and $s, t \in V$ where there exists $C^* \subseteq V \setminus \{s, t \}$ such that $\dist(G \setminus C^*) \geq l$, it is NP-hard to perform any of the following tasks.
\begin{enumerate}
\item Find $C \subseteq V \setminus \{ s, t \}$ such that $|C| \leq \Omega(l) \cdot |C^*|$ and $\dist(G \setminus C) \geq l$.
\item Find $C \subseteq V \setminus \{ s, t \}$ such that $|C| \leq |C^*|$ and $\dist(G \setminus C) \geq O(\sqrt{l})$.
\item Find $C \subseteq V \setminus \{ s, t \}$ such that $|C| \leq \Omega(l^{\frac{\epsilon}{2}}) \cdot |C^*|$ and $\dist(G \setminus C) \geq O(l^{\frac{1 + \epsilon}{2}})$ for some $0 < \epsilon < 1$.
\end{enumerate}
\label{thm:vertex1}
\end{thm}

The first result shows that Length Bounded Vertex Cut is hard to approximate within a factor $\Omega(l)$. This matches the best $\frac{l}{2}$-approximation up to a constant.~\cite{BEHKKPSS10}. 
The second result shows that Shortest Path Vertex Interdiction is hard to approximate with in a factor $\Omega(\sqrt{\OPT})$, and the third result rules out {\em bicriteria approximation} --- for any constant $c$, it is hard to approximate both $l$ and $|C^*|$ within a factor of $c$. 

The above results hold for directed graphs by definition. Our hard instances will have a natural {\em layered} structure, so it can be easily checked that the same results (up to a constant) hold for directed acyclic graphs. Since one vertex can be split as one directed edge, the same results hold for the edge version in directed acyclic graphs. 

For Length-Bounded Edge Cut and Shortest Path Edge Interdiction in undirected graphs (collectively called {\em Short Path Edge Cut} onwards), we prove the following theorems. 

\begin{thm}
Assume the Unique Games Conjecture. 
For infinitely many values of the constant $l \in \NN$, 
given an undirected graph $G = (V, E)$ and $s, t \in V$ where there exists $C^* \subseteq E$ such that $\dist(V \setminus C^*) \geq l$, it is NP-hard to perform any of the following tasks.
\begin{enumerate}
\item Find $C \subseteq E$ such that $|C| \leq \Omega(\sqrt{l}) \cdot |C^*|$ and $\dist(G \setminus C) \geq l$.
\item Find $C \subseteq E$ such that $|C| \leq |C^*|$ and $\dist(G \setminus C) \geq l^{\frac{2}{3}}$.
\item Find $C \subseteq E$ such that $|C| \leq \Omega(l^{\frac{2\epsilon}{3}}) \cdot |C^*|$ and $\dist(G \setminus C) \geq O(l^{\frac{2 + 2\epsilon}{3}})$ for some $0 < \epsilon < \frac{1}{2}$.
\end{enumerate}
\label{thm:edge1}
\end{thm}

Our hardness factors for the edge versions, $\Omega(\sqrt{l})$ for Length-Bounded Edge Cut and $\Omega(\sqrt[3]{\OPT})$ for Shortest Path Edge Interdiction, are slightly weaker than those for their vertex counterparts, but we are not aware of any approximation algorithm specialized for the edge versions. It is an interesting open problem whether there exist better approximation algorithms for the edge versions. 

\paragraph{RMFC.}
{\em Resource Minimization for Fire Containment (RMFC)} is a problem closely related to Length-Bounded Cut with the additional notion of time. Given a graph $G$, a vertex $s$, and a subset $T$ of vertices, consider the situation where fire starts at $s$ on Day $0$. For each Day $i$ ($i \geq 1$), we can {\em save} at most $k$ vertices, and the fire spreads from currently burning vertices to its unsaved neighbors. 
Once a vertex is burning or saved, it remains
so from then onwards. The process is terminated when the fire cannot spread anymore. 
RMFC asks to find a strategy to save $k$ vertices each day with the minimum $k$ so that no vertex in $T$ is burnt. These problems model the spread of epidemics or ideas through a social network, and have been actively studied recently~\cite{CC10, ACHS12, ABZ16, CV16}.

RMFC, along with other variants, is first introduced by Hartnell~\cite{Hartnell95}. 
Another well-studied variant is called the {\em Firefighter} problem, where we are only given $s \in V$ and want to maximize the number of vertices that are not burnt at the end. It is known to be NP-hard to approximate within a factor $n^{1 - \epsilon}$ for any $\epsilon > 0$~\cite{ACHS12}.
King and MacGillivray~\cite{KM10} proved that RMFC is hard to approximate within a factor less than $2$. Anshelevich et al.~\cite{ACHS12} presented an $O(\sqrt{n})$-approximation algorithm for general graphs, and Chalermsook and Chuzhoy~\cite{CC10} showed that RMFC admits $O(\log^* n)$-approximation in trees. Very recently, the approximation ratio in trees has been improved to $O(1)$~\cite{ABZ16}. Both Anshelevich et al.~\cite{ACHS12} and Chalermsook and Chuzhoy~\cite{CC10} independently studied directed layer graphs with $b$ layers, showing $O(\log b)$-approximation.

Our final result on RMFC assumes Conjecture~\ref{conj:ug2}, 
a variant of the Unique Games Conjecture which is not known to be equivalent to the original UGC. Given a bipartite graph as an instance of Unique Games, it states that in the completeness case, all constraints incident on $(1 - \epsilon)$ fraction of vertices in one side are satisfied, and in the soundness case, in addition to having a low value, every $\frac{1}{10}$ fraction of vertices on one side have at least a $\frac{9}{10}$ fraction of vertices on the other side as neighbors. Our conjecture is implied by the conjecture of Bansal and Khot~\cite{BK09} that is used to prove the hardness of {\em Minimizing Weighted Completion Time with Precedence Constraints} and requires a more strict expansion condition.
See Section~\ref{sec:ug} for the exact statement. 

\begin{thm}
Assuming Conjecture~\ref{conj:ug2}, it is NP-hard to approximate RMFC in undirected graphs within any constant factor. 
\label{thm:firefighter}
\end{thm}

Again, our reduction has a natural layered structure and the result holds for directed layered graphs. With $b$ layers, we prove that it is hard to approximate with in a factor $\Omega(\log b)$, matching the best approximation algorithms~\cite{CC10, ACHS12}.

\subsection{Techniques}
\label{subsec:techniques}
All our results are based on a general method of {\em converting an integrality gap instance to a dictatorship test}. This method has been successfully applied by Raghavendra~\cite{Raghavendra08} for Max-CSPs, 
Manokaran et al.~\cite{MNRS08} and Ene et al.~\cite{EVW13} for Multiway Cut and Min CSPs, and Kumar et al.~\cite{KMTV11} for strict CSPs, and by Guruswami et al.~\cite{GSS15} for $k$-uniform $k$-partite Hypergraph Vertex Cover. 
As mentioned in the introduction, the previous CSP-based results do not generally preserve the structure of constraint hypergraphs or use ingenious and specialized tricks to reduce the problem to a CSP,  so they are not applicable as a black-box to the graph cut problems we consider.

We bypass this difficulty by constructing a special class of dictatorship tests that we call {\em length-control dictatorship tests}. 
Consider a meta-problem where given a directed graph $G = (V, E)$, some terminal vertices, and a set $\mathcal{P}$ of {\em desired paths} between terminals, we want to remove the fewest number of non-terminal vertices to cut every path in $\mathcal{P}$.
The integrality gap instances we use in this work~\cite{SSZ04, BEHKKPSS10, MM10, CC10} share the common feature that every $p \in \mathcal{P}$ is of length at least $r$, and the fractional solution cuts $\frac{1}{r}$ fraction of each non-terminal vertex so that each path $p \in \mathcal{P}$ is cut. This gives a good LP value, and additional arguments are required to ensure that there is no efficient integral cut.

Given such an integrality gap instance, we construct our dictatorship test instance as follows. 
We replace every non-terminal vertex by a {\em hypercube} $\ZZ_r^R$ and put edges such that for two vertices $(v, x)$ and $(w, y)$ where $v, w \in V$ and $x, y \in \ZZ_r^R$, there is an edge from $(v, x)$ to $(w, y)$ if (1) $(v, w) \in E$ and (2) $y_j = x_j + 1$ for all $j \in [R]$. 
The set of desired paths $\mathcal{P'}$ is defined to be $\{ (s, (v_1, x_1), \dots, (v_l, x_l), t) : (s, v_1, \dots, v_l, t) \in \mathcal{P} \}$ ($s, t$ denote some terminals). Note that each path in $\mathcal{P'}$ is also of length at least $r$.
We want to ensure that in the {\em completeness case} (i.e., every hypercube reveals the same {\em influential coordinate}), there is a very efficient cut, while in the {\em soundness case} (i.e., no hypercube reveals an influential coordinate), there is no such efficient cut.

In the completeness case, let $q \in [R]$ be an influential coordinate. 
For each vertex $(v, x)$ where $v \in V, x \in \ZZ_r^R$, remove $(v, x)$ if $x_{q} = 0$. 
Consider a desired path $p = (s, (v_1, x_1), \dots, (v_l, x_l), t) \in \mathcal{P'}$ for some terminals $s, t$ and some $v_j \in V, x_j \in \ZZ_r^R \, \, (1 \leq j \leq l)$, and let $y_j = (x_j)_q$. By our construction, $y_{j + 1} = y_j + 1$ for $0 \leq j < l$. 
Since $p$ is desirable, $l \geq r$, so there exists $j$ such that $y_j = (x_j)_q = 0$, but $(v_j, x_j)$ is already removed by our previous definition. Therefore, every desired path is cut by this vertex cut. Note that this cut is integral and cuts exactly $\frac{1}{r}$ fraction of non-terminal vertices. This corresponds to the fractional solution to the gap instance that cuts $\frac{1}{r}$ fraction of every vertex.

For the soundness analysis, our final dictatorship test has additional {\em noise} vertices and edges to the test defined above. If no hypercube reveals an influential coordinate, the standard application of the invariance principle~\cite{Mossel10} proves that we can always take an edge between two hypercubes unless we almost completely cut one hypercube. We can then invoke the proof for the integrality gap instance to show that there is no efficient cut.

This idea is implicitly introduced by the work of Svensson~\cite{Svensson13} for Feedback Vertex Set (FVS) and DAG Vertex Deletion (DVD) by applying the {\em It ain't over till it's over theorem} to ingeniously constructed dictatorship tests with auxiliary vertices. Guruswami and Lee~\cite{GL16} gave a simpler construction and a new proof using the invariance principle instead of the It ain't over till it's over theorem. Our results are based on the observation that length-control dictatorship tests and LP gap instances {\em fool} algorithms in a similar way for various cut problems as mentioned above, so that the previous LP gap instances can be plugged into our framework to prove matching hardness results.

This method for the above meta-problem can be almost directly applied to Directed Multicut. For Length-Bounded Cut and RMFC in undirected graphs, we use the fact that the known integrality gap instances have a natural layered structure with $s$ in the first layer and $t$ in the last layer. Every edge is given a natural orientation, and the similar analysis can be applied. 
For Length-Bounded Cut, another set of edges called {\em long edges} are added to the dictatorship test. More technical work is required for edge cut versions in undirected graphs (\edgecutn), and the notion of time (RMFC). 

Our framework seems general enough so that they can be applied to integrality gap instances to give strong hardness results. Since each problem has slightly different characteristics as mentioned above, each application needs some specialized ideas and sometimes leads to sub-optimal results. It would be interesting to further abstract this method of converting integrality gap instances to length-bounded dictatorship tests, as well as to apply it to other problems whose approximability is not well-understood.

\section{Preliminaries}
\paragraph{Graph Terminologies.}
Depending on whether we cut vertices or edges, we introduce weight $\w(v)$ for each vertex $v$, or weight $\w(e)$ for each edge $e$. 
Some weights can be $\infty$, which means that some vertices or edges cannot be cut. For vertex-weighted graphs, we naturally have $\w(s) = \w(t) = \infty$. 
To reduce the vertex-weighted version to the unweighted version, we duplicate each vertex according to its weight and replace each edge by a complete bipartite graph between corresponding copies.
To reduce the edge-weighted version to the unweighted version, we replace a single edge with parallel edges according to its weight. To reduce to simple graphs, we split each parallel into two edges by introducing a new vertex. 

For the Length-Bounded Cut problems, we also introduce length $\len(e)$ for each edge $e$. 
It can be also dealt with serially splitting an edge according to its weight. 
We allow weights to be rational numbers, but as our hardness results are stated in terms of the length, all lengths in this work will be a positive integer. 

For a path $p$, depending on the context, we abuse notation and interpret it as a set of edges or a set of vertices. The length of $p$ is always defined to be the number of edges.

\paragraph{Gaussian Bounds for Correlated Spaces.}
We introduce the standard tools on correlated spaces from Mossel~\cite{Mossel10}.
Given a probability space $(\Omega, \mu)$ (we always consider finite probability spaces), let $\mathcal{L}(\Omega)$ be the set of functions $\left\{ f : \Omega \rightarrow \RR \right\}$ and for an interval $I \subseteq \RR$, 
$\mathcal{L}_I (\Omega)$ be the set of functions $\left\{ f : \Omega \rightarrow I \right\}$. For a subset $S \subseteq \Omega$, define {\em measure} of $S$ to be $\mu(S) := \sum_{\omega \in S} \mu(\omega)$. 
A collection of probability spaces are said to be correlated if there is a joint probability distribution on them. We will denote $k$ correlated spaces $\Omega_1,\dots,\Omega_k$ with a joint distribution $\mu$ as $(\Omega_1 \times \cdots \times \Omega_k , \mu)$.

Given two correlated spaces $(\Omega_1 \times \Omega_2, \mu)$, we define the correlation between $\Omega_1$ and $\Omega_2$ by 
\[
\rho(\Omega_1, \Omega_2; \mu) := \sup \left\{ \Cov[f, g] : f \in \mathcal{L}(\Omega_1), g \in \mathcal{L}(\Omega_2), \Var[f] = \Var[g] = 1 \right\}.
\]
Given a probability space $(\Omega, \mu)$ and a function $f \in \mathcal{L}(\Omega)$ and $p \in \RR^+$, let $\|f\|_{p} := \E_{x \sim \mu}[|f(x)|^p]^{1/p}$. 

Consider a product space $(\Omega^R, \mu^{\otimes R})$ and $f \in \mathcal{L}(\Omega^R)$. 
The {\em Efron-Stein decomposition} of $f$ is given by
\[
f(x_1, \dots , x_R) = \sum_{S \subseteq [R]} f_S(x_S)
\]
where (1) $f_S$ depends only on $x_S$ and (2) 
for all $S \not \subseteq S'$ and all $x_{S'}$, $
\E_{x' \sim \mu^{\otimes R}} [f_S(x') | x'_{S'} = x_{S'}] = 0$.
The {\em influence} of the $i$th coordinate on $f$ is defined by 

\[
\quad \Inf_i[f] := \E_{x_1, \dots , x_{i - 1}, x_{i + 1}, \dots , x_R} [\Var_{x_i} [f(x_1, \dots , x_R)].
\]
The influence has a convenient expression in terms of the Efron-Stein decomposition. 
\[
\Inf_i[f]  = \| \sum_{S : i \in S} f_S \|_2^2 = 
\sum_{S : i \in S} \| f_S \|_2^2.
\]
We also define the {\em low-degree influence} of the $i$th coordinate.
\[
\mathsf{Inf}_i^{\leq d}[f] := 
\sum_{S : i \in S, |S| \leq d} \| f_S \|_2^2.
\]
For $a, b \in [0, 1]$ and $\rho \in (0, 1)$, let 
\[
\Gamma_{\rho}(a, b) := \Pr[X \leq \Phi^{-1}(a), Y \geq \Phi^{-1}(1 - b)],
\]
where $X$ and $Y$ are $\rho$-correlated standard Gaussian variables and $\Phi$ denotes the cumulative distribution function of a standard Gaussian. 
The following theorem bounds the product of two functions that do not share an influential coordinate in terms of their Gaussian counterparts. 

\begin{thm}[Theorem 6.3 and Lemma 6.6 of~\cite{Mossel10}]
Let $(\Omega_1 \times \Omega_2, \mu)$ be correlated spaces such that
the minimum nonzero probability of any atom in $\Omega_1 \times \Omega_2$ is at least $\alpha$ and such
that $\rho(\Omega_1, \Omega_2; \mu) \leq \rho$. 
Then for every $\epsilon > 0$ there exist $\tau, d$ depending on $\epsilon$ and $\alpha$ such that if $f : \Omega_1^R \rightarrow [0, 1], g : \Omega_2^R \rightarrow [0, 1]$ satisfy
$
\min(\Inf_i^{\leq d}[f], \Inf_i^{\leq d}[g]) \leq \tau
$
for all $i$, then
$
\E_{(x, y) \in \mu^{\otimes R}} [f(x)g(y)] \geq \Gamma_{\rho}(\E_x[f], \E_y [g]) - \epsilon.
$
\label{thm:mossel}
\end{thm}

\paragraph{Organization.}
Section~\ref{sec:test_multicut} shows the dictatorship tests for Directed Multicut.
We present our dictatorship tests for \edgecut and \vertexcut in Section~\ref{sec:test_edge} and Section~\ref{sec:test_vertex} respectively. The dictatorship tests for RMFC are presented in Section~\ref{sec:test_rmfc}. 
These tests will be used in Section~\ref{sec:ug} to prove hardness results based on the UGC.

\section{Directed Multicut}
\label{sec:test_multicut}
We propose our dictatorship test for Directed Vertex Multicut that will be used for proving Unique Games hardness. 
Note that hardness of Directed Edge Multicut easily follows from that of the vertex version by splitting each vertex.
Our dictatorship test is inspired by the integrality gap for the standard LP constructed by Saks et al.~\cite{SSZ04}, and parameterized by positive integers $r, k, R$ and small $\epsilon > 0$, where $k$ in this section denotes the number of $(s_i, t_i)$ pairs for Directed Multicut. All graphs in this section are directed. 

For positive integers $r, k, R$, and $\epsilon > 0$, define $\calD^{\SM}_{r, k, R, \epsilon} = (V, E)$ be the graph defined as follows. Consider the probability space $(\Omega, \mu)$ where $\Omega := \{ 0, \dots, r-1, * \}$, and $\mu : \Omega \mapsto [0, 1]$ with $\mu(*) = \epsilon$ and $\mu(x) = \frac{1 - \epsilon}{r}$ for $x \neq *$. 

\begin{itemize}
\item $V = \{ s_i, t_i \}_{1 \leq i \leq k} \cup \{ v^{\alpha}_x \}_{\alpha \in [r]^k, x \in \Omega^R}$. Let $v^{\alpha}$ denote the set of vertices $\{ v^{\alpha}_x \}_{x \in \Omega^R}$. 

\item For $\alpha \in [r]^k$ and $x \in \Omega^R$, $\w(v^i_x) = \mu^{\otimes R}(x)$. Note that the sum of weights is $r^k$. 

\item For any $i \in [k]$, there are edges from $s_i$ to $\{ v^{\alpha}_x: \alpha \in [r]^k, \alpha_i = 1, x \in \Omega^R \}$, and edges from $\{ v^{\alpha}_x: \alpha \in [r]^k, \alpha_i = r, x \in \Omega^R \}$ to $t_i$. 

\item For $\alpha, \beta \in [r]^k$ and $x, y \in \Omega^R$, we have an edge from $v^{\alpha}_x$ to $v^{\beta}_y$ if $\alpha \neq \beta$ and 
\begin{itemize}
\item For any $1 \leq i \leq r$: $\alpha_i - \beta_i \in \{ -1, 0, +1 \}$. 
\item For any $1 \leq j \leq R$: [$y_j = (x_j + 1)\mod r$] or [$y_j = *$] or [$x_j = *$]. 
\end{itemize}
\end{itemize}

\paragraph{Completeness.} 
We first prove that vertex cuts that correspond to {\em dictators} behave the same as the fractional solution that gives $\frac{1}{r}$ to every vertex. For any $q \in [R]$, let $V_q := \{ v^{\alpha}_x : \alpha \in [r]^k, x_q = * \mbox{ or } 0 \}$. Note that the total weight of $V_q$ is $r^k(\epsilon + \frac{1 - \epsilon}{r}) \leq r^{k - 1}(1 + \epsilon r)$. 

\begin{lem}
After removing vertices in $V_q$, there is no path from $s_i$ to $t_i$ for any $i$.
\label{lem:dict_multicut}
\end{lem}
\begin{proof}
Fix $i$ and let $p = (s_i, v^{\alpha_1}_{x^1}, \dots, v^{\alpha_z}_{x^z}, t_i)$ be a path from $s_i$ to $t_i$
where $\alpha_j \in [r]^k$ and $x^j \in \Omega^R$ for each $1 \leq j \leq z$. Let $y_j := (x^j)_q$ for each $1 \leq j \leq z$. 
The construction ensures that $y_{j + 1} = (y_j + 1) \mod r$, 
so after removing vertices in $V_q$, $z$ must be strictly less than $r$. Since any path from $s_i$ to $t_i$ must contain at least $r$ non-terminal vertices, there must be no path from $s_i$ to $t_i$.
\end{proof}

\paragraph{Soundness.}
To analyze soundness, we define a correlated probability space $(\Omega_1 \times \Omega_2, \nu)$ where both $\Omega_1, \Omega_2$ are copies of $\Omega = \{ 0, \dots, r - 1, * \}$. It is defined by the following process to sample $(x, y) \in \Omega^2$.
\begin{itemize}
\item Sample $x \in \{0, \dots, r - 1 \}$. Let $y = (x + 1) \mod r$. 
\item Change $x$ to $*$ with probability $\epsilon$. Do the same for $y$ independently. 
\end{itemize}
Note that the marginal distribution of both $x$ and $y$ is equal to $\mu$. Assuming $\epsilon < \frac{1}{2r}$, the minimum probability of any atom in $\Omega_1 \times \Omega_2$ is $\epsilon^2$. We use the following lemma to bound the correlation $\rho(\Omega_1, \Omega_2; \nu)$. 

\begin{lem}[Lemma 2.9 of~\cite{Mossel10}]
Let $(\Omega_1 \times \Omega_2, \mu)$ be two correlated spaces such that the probability of the smallest atom in $\Omega_1 \times \Omega_2$ is at least $\alpha > 0$. Define a bipartite graph $G = (\Omega_1 \cup \Omega_2, E)$ where $(a, b) \in \Omega_1 \times \Omega_2$ satisfies $(a, b) \in E$ if $\mu(a, b) > 0$. If $G$ is connected, then
$
\rho(\Omega_1, \Omega_2 ; \mu) \leq 1 - \frac{\alpha^2}{2}.
$
\label{lem:mossel_connected}
\end{lem}
In our correlated space, the bipartite graph on $\Omega_1 \cup \Omega_2$ is connected since every $x \in \Omega_1$ is connected to $* \in \Omega_2$ and vice versa. Therefore, we can conclude that 
$\rho(\Omega_1, \Omega_2 ; \nu) \leq \rho := 1 - \frac{\epsilon^4}{2}$.

Apply Theorem~\ref{thm:mossel} ($\rho \leftarrow \rho, \alpha \leftarrow \epsilon^2, \epsilon \leftarrow \frac{\Gamma_{\rho}(\frac{\epsilon}{3}, \frac{\epsilon}{3})}{2} $) to get $\tau$ and $d$. We will later apply this theorem with the parameters obtained here. 
Fix an arbitrary subset $C \subseteq V$, and let $C_{\alpha}:= C \cap v^{\alpha}$.  
For $\alpha \in [r]^k$, call $v^{\alpha}$ {\em blocked} if $\mu^{\otimes R} (C_{\alpha}) \geq 1 - \epsilon$. The number of blocked $v^{\alpha}$'s is at most $\frac{\w(C)}{1 - \epsilon}$.

Consider the following graph $D = (V_D, E_D)$, which is the original integrality gap instance constructed by Saks et al.~\cite{SSZ04}. 
\begin{itemize}
\item $V_D = \{ s_i, t_i \}_{i \in [k]} \cup \{ v^{\alpha} \}_{\alpha \in [r]^k}$.

\item For any $i \in [k]$, there are edges from $s_i$ to $\{ v^{\alpha}: \alpha \in [r]^k, \alpha_i = 1 \}$, and edges from $\{ v^{\alpha}: \alpha \in [r]^k, \alpha_i = r \}$ to $t_i$. 

\item For $\alpha, \beta \in [r]^k$, we have an edge from $v^{\alpha}$ to $v^{\beta}$ if $\alpha \neq \beta$ and $1 \leq i \leq r$: $\alpha_i - \beta_i \in \{ -1, 0, +1 \}$. 
\end{itemize}
Saks et al.~\cite{SSZ04} proved the following theorem in their analysis of their integrality gap. 

\begin{thm}
Let $C'$ be a set of less than $k(r - 1)^{k - 1}$ vertices. There exists a path $(s_i, v^{\alpha_1}, \dots, v^{\alpha_z}, t_i)$ for some $i$ that does not intersect $C'$. 
\label{thm:ssz}
\end{thm}
Setting $C' := \{ v^{\alpha} \in V_D : v^{\alpha} \mbox{ is not blocked.} \}$, and applying Theorem~\ref{thm:ssz} concludes that unless $\w(C) \geq (1 - \epsilon) \cdot k \cdot (r - 1)^{k - 1}$, there exists a path $(s_i, v^{\alpha_1}, \dots, v^{\alpha_z}, t_i)$ where each $v^{\alpha_i}$ is unblocked for some $i \in [k]$. 

For $1 \leq j \leq z$, let $S_j \subseteq v^{\alpha_j}$ be such that $x \in S_j$ if there exists a path $(s, v^{\alpha_1}_{x^1}, \dots,
v^{\alpha_{j - 1}}_{x^{j - 1}}, v^{\alpha_j}_{x})$ for some $x^1, \dots, x^{j - 1}$. 
For $1 \leq j \leq z$, let $f_j : \Omega^R \mapsto \{ 0, 1 \}$ be the indicator function of $S_j$. 
We prove that if none of $f_j$ reveals any {\em influential coordinate}, $\mu^{\otimes R} (S_z) > 0$, which shows that there exists a $s_i$-$t_i$ path even after removing vertices in $C$.

\begin{lem}
Suppose that for any $1 \leq j \leq z$ and $1 \leq i \leq R$, $\Inf_i^{\leq d}[f_j] \leq \tau$. Then $\mu^{\otimes R}(S_z) > 0$. 
\end{lem}
\begin{proof}
We prove by induction that $\mu^{\otimes R} (S_j) \geq \frac{\epsilon}{3}$. It holds when $j = 1$ since $v^{\alpha_1}$ is unblocked. Assuming $\mu^{\otimes R} (S_j) \geq \frac{\epsilon}{3}$, since $S_j$ does not reveal any influential coordinate, Theorem~\ref{thm:mossel} shows that for any subset $T_{j + 1} \subseteq v^{\alpha_{j + 1}}$ with $\mu^{\otimes R}(T_{j + 1}) \geq \frac{\epsilon}{3}$, there exists an edge from $S_j$ and $T_{j + 1}$. If $S'_{j + 1} \subseteq v^{\alpha_{j + 1}}$ is the set of out-neighbors of $S_j$, we have $\mu^{\otimes R} (S'_{j + 1}) \geq 1 - \frac{\epsilon}{3}$. Since $v^{\alpha_{j + 1}}$ is unblocked, $\mu^{\otimes R} (S'_{j + 1} \setminus C) \geq \frac{2\epsilon}{3}$, completing the induction. 
\end{proof}

In summary, in the completeness case, if we cut vertices of total weight $r^{k - 1}(1 + \epsilon r)$, we cut every $s_i$-$t_i$ pair. 
In the soundness case, unless we cut vertices of total weight at least $(1 - \epsilon) \cdot k \cdot (r - 1)^{k - 1}$, we cannot cut every $s_i$-$t_i$ pair. The gap is 
$\frac{k (1 - \epsilon)  (r - 1)^{k - 1}}{(1 + \epsilon r) r^{k - 1}}$. For a fixed $k$, increasing $r$ and decreasing $\epsilon$ faster makes the gap arbitrarily close to $k$.

\section{Short Path Edge Cut}
\label{sec:test_edge}
We propose our dictatorship test for \edgecut that will be used for proving Unique Games hardness. 
It is parameterized by positive integers $a, b, r, R$. 
It is inspired by the integrality gap instances by Baier et al.~\cite{BEHKKPSS10} Mahjoub and and McCormick~\cite{MM10}, and made such that the edge cuts that correspond to {\em dictators} behave the same as the fractional solution that cuts $\frac{1}{r}$ fraction of every edge. All graphs in this section are undirected.

For positive integers $a, b, r, R$, we construct $\calD^{\SE}_{a, b, r, R} = (V, E)$. Let $\Omega = \{ 0, \dots, r-1 \}$, and $\mu : \Omega \mapsto [0, 1]$ with $\mu(x) = \frac{1}{r}$ for each $x \in \Omega$. 
We also define a correlated probability space $(\Omega_1 \times \Omega_2, \nu)$ where both $\Omega_1, \Omega_2$ are copies of $\Omega$. It is defined by the following process to sample $(x, y) \in \Omega^2$.
\begin{itemize}
\item Sample $x \in \{0, \dots, r - 1 \}$. Let $y = (x + 1) \mod r$. 
\item With probability $1 - \frac{1}{r}$, output $(x, y)$. Otherwise, resample $x, y \in \Omega$ independently and output $(x, y)$. 
\end{itemize}
Note that the marginal distribution of both $x$ and $y$ is equal to $\mu$. Given $x = (x_1, \dots, x_R) \in \Omega^R$ and $y = (y_1, \dots, y_R) \in \Omega^R$, let $\nu^{\otimes R}(x, y) = \prod_{i = 1}^R \nu(x_i, y_i)$. 
We define $\calD^{\SE}_{a, b, r, R} = (V, E)$ as follows. 

\begin{itemize}
\item $V = \{ s, t \} \cup \{ v^i_x \}_{0 \leq i \leq b, x \in \Omega^R}$. Let $v^i$ denote the set of vertices $\{ v^i_x \}_{x \in \Omega^R}$. 

\item For any $x \in \Omega^R$, there is an edge from $s$ to $v^0_x$ and an edge from $v^b_x$ to $t$, both with weight $\infty$ and length $1$. 

\item For $0 \leq i < b, x \in \Omega^R$, there is an edge $(v^i_x, v^{i+1}_x)$ of length $a$ and weight $\infty$. Call it a {\em long edge}. 

\item For any $0 \leq i < b$ $x, y \in \Omega^R$, there is an edge $(v^i_x, v^{i+1}_y)$ of length $1$ and weight $\nu^{\otimes R}(x, y)$. Note that $\nu^{\otimes R}(x, y) > 0$ for any $x, y \in \Omega^R$. Call it a {\em short edge}. 
The sum of finite weights is $b$. 
\end{itemize}

\paragraph{Completeness.} 
We first prove that edge cuts that correspond to {\em dictators} behave the same as the fractional solution that gives $\frac{1}{r}$ to every edge. Fix $q \in [R]$ and let $E_q$ be the set of short edges defined by
\[
E_q := \{ (v^i_x, v^{i + 1}_y) : 0 \leq i < b, y_q \neq x_q + 1 \mod R \mbox{ or } (x_q, y_q) = (0, 1) \}. 
\]
When $(x, y) \in \Omega_1 \times \Omega_2$ is sampled according to $\nu$, the probability that 
$y_q \neq x_q + 1 \mod R \mbox{ or } (x_q, y_q) = (0, 1)$ is at most $\frac{2}{r}$. 
The total weight of $E_q$ is $\frac{2b}{r}$. 

\begin{lem}
After removing edges in $E_q$, the length of the shortest path is at least $a(b - r + 1)$.
\label{lem:dict_edge}
\end{lem}
\begin{proof}
Let $p = (s, v^{i_1}_{x^1}, \dots, v^{i_z}_{x^z}, t)$ be a path from $s$ to $t$ where $i_j \in \{ 0, \dots, b \}$ and $x^j \in \Omega^R$ for each $1 \leq j \leq z$. Let $y_j := (x^j)_q \in \{ 0, \dots, r - 1 \}$ for each $1 \leq j \leq z$. 

For each $1 \leq j < z$, the edge $(p_j, p_{j + 1})$ is either a long edge or a short edge, and either taken forward (i.e., $i_j < i_{j + 1}$) or backward (i.e., $i_j > i_{j + 1}$). Let $z_{\SLF}, z_{\SSF}, z_{\SLB}, z_{\SSB}$ be the number of long edges taken forward, short edges taken forward, long edges taken  backward, and shot edges taken backward, respectively ($z_{\SLF} + z_{\SSF} + z_{\SLB} + z_{\SSB} = z - 1$). By considering how $i_j$ changes, 
\begin{equation}
z_{\SLF} + z_{\SSF} - z_{\SLB} - z_{\SSB} = b.
\label{eq:edge1}
\end{equation}
Consider how $y_j$ changes. Taking a long edge does not change $y_j$. Taking a short edge forward increases $y_j$ by $1$ mod $r$, taking a short edge backward decreases $y_j$ by $1$ mod $r$. 
Since $E_q$ is cut, $y_j$ can never change from $0$ to $1$. This implies 
\begin{equation}
z_{\SSF} - z_{\SSB} \leq r - 1.
\label{eq:edge2}
\end{equation}
$\eqref{eq:edge1} - \eqref{eq:edge2}$ yields $z_{\SLF} - z_{\SLB} \geq b - r + 1$.
The total length of $p$ is at least $a \cdot z_{\SLF} \geq a(b - r + 1)$.
\end{proof}

\paragraph{Soundness.}
We first bound the correlation $\rho(\Omega_1, \Omega_2 ; \nu)$. 
The following lemma of Wenner~\cite{Wenner13} gives a convenient way to bound the correlation.

\begin{lem}[Corollary 2.18 of \cite{Wenner13}]
\label{lem:correlation_bound}
Let $(\Omega_1 \times \Omega_2, \delta \mu + (1 - \delta)\mu')$ be two correlated spaces such that the marginal distribution of at least one of $\Omega_1$ and $\Omega_2$ is identical on $\mu$ and $\mu'$. Then,
\[
\rho(\Omega_1, \Omega_2; \delta \mu + (1 - \delta)\mu') \leq \sqrt{\delta \cdot \rho(\Omega_1, \Omega_2; \mu)^2 + (1 - \delta) \cdot \rho(\Omega_1, \Omega_2; \mu')^2}.
\]
\end{lem}
When $(x, y)$ is sampled from $\nu$, they are completely independent with probability $\frac{1}{r}$. Therefore, we have $\rho := \rho(\Omega_1, \Omega_2 ; \nu) \leq \sqrt{1 - \frac{1}{r}}$. 
By Sheppard's Formula, 
\[
\Gamma_{\rho}(\frac{1}{2}, \frac{1}{2})  = 
\frac{1}{4} + \frac{1}{2 \pi} \arcsin (-\rho) \geq 
\frac{1}{4} - \frac{1}{2 \pi} \arccos (\frac{1}{\sqrt{r}}) = 
\sum_{n = 0}^{\infty} 
\frac{(2n)!}{4^n (n!)^2 (2n + 1)} (\frac{1}{\sqrt{r}})^{2n + 1} \geq \frac{1}{\sqrt{r}}.
\]

Apply Theorem~\ref{thm:mossel} ($\rho \leftarrow \rho, \alpha \leftarrow \frac{1}{r^3}, \epsilon \leftarrow \frac{\Gamma_{\rho}(\frac{1}{2}, \frac{1}{2})}{3} $) to get $\tau$ and $d$. We will later apply this theorem with the parameters obtained here. 

Fix an arbitrary subset $C \subseteq E$ of short edges.
For $0 \leq i < b$, let $C_i = C \cap (v^i \times v^{i+1})$. Call a pair $(i, i + 1)$ as the $i$th layer, and say it is blocked when $\nu^{\otimes R}(C_i) \geq \frac{\Gamma_{\rho}(\frac{1}{2}, \frac{1}{2})}{2}$. 
Let $b'$ be the number of blocked layers. 
For $0 \leq i \leq b$, let $S_i \subseteq v^{i}$ be such that $x \in S_i$ if there exists a path $(s, p_0, \dots, p_i = v^i_x)$ such that
\begin{itemize}
\item For $0 \leq i' \leq i$, $p_{i'} \in v^{i'}$. 
\item For $0 \leq i' < i$, $(p_{i'}, p_{i' + 1})$ is short if and only if the $i'$th layer is unblocked. 
\end{itemize}

Let $f_i : \Omega^R \mapsto [ 0, 1 ]$ be the indicator function of $S_i$. 
We prove that if none of $f_i$ reveals any {\em influential coordinate}, $S_b$ is nonempty, implying that there exists a path using $b'$ long edges and $b - b'$ short edges. 
. Therefore, even after removing edges in $C$, the length of the shortest path is at most $2 + ab' + (b - b')$. 

\begin{lem}
Suppose that for any $0 \leq i \leq b$ and $1 \leq j \leq R$, $\Inf_j^{\leq d}[f_i] \leq \tau$. Then $S_b \neq \emptyset$.
\end{lem}
\begin{proof}
Assume towards contradiction that $S_b = \emptyset$. 
Since $S_0 = \Omega^R$ and $S_{i} = S_{i + 1}$ if the $i$th layer is blocked (and we use long edges), there must exist $i$ such that the $i$th layer is unblocked and $\mu^{\otimes R}(S_i) \geq \frac{1}{2}, \mu^{\otimes R}(S_{i + 1}) < \frac{1}{2}$. All short edges between $S_{i}$ and $v^{i + 1} \setminus S_{i + 1}$ are in $C_i$.  Theorem~\ref{thm:mossel} implies that $\nu^{\otimes R}(C_i) > \frac{2}{3} \Gamma_{\rho}(\frac{1}{2}, \frac{1}{2})$. This contradicts the fact that the $i$th layer is unblocked. 
\end{proof}

In summary, in the completeness case, if we cut edges of total weight $k := k(a, b, r) = \frac{2b}{r}$, the length of the shortest path is at least $l := l(a, b, r) = a(b - r + 1)$. 
In the soundness case, even after cutting edges of total weight $k'$, at most $\frac{2k'}{\Gamma_{\rho}(\frac{1}{2}, \frac{1}{2})} \leq 2k'\sqrt{r}$ layers are blocked, 
the length of the shortest path is at most $l' = 2 + (b - 2k' \sqrt{r}) + 2 a k' \sqrt{r}$. 

\begin{itemize}
\item Let $a = 4, b = 2r - 1$ so that $k \leq 4, l = 4r$. Requiring $l' \geq l$ results in $k' = \Omega(\sqrt{r})$, giving a gap of $\Omega(\sqrt{r}) = \Omega(\sqrt{l})$  between the completeness case and the soundness case for Length-Bounded Edge Cut.

\item Let $a = \sqrt{r}, b = 2r - 1$ so that $k \leq 4, l = r^{1.5}$. Requiring $k' \leq 4$ results in $l' = O(r)$, giving a gap of $\Omega(\sqrt{r}) = \Omega(l^{1/3})$ for Shortest Path Interdiction. Generally, $k' \leq O(r^{\epsilon})$ results in $l' \leq O(r^{1 + \epsilon})$, giving an $(O(r^{\epsilon}), O(r^{1/2 - \epsilon}))$-bicriteria gap for any $\epsilon \in (0, \frac{1}{2})$. 
\end{itemize}

\section{Short Path Vertex Cut}
\label{sec:test_vertex}
We propose our dictatorship test for \vertexcut that will be used for proving Unique Games hardness. 
It is parameterized by positive integers $a, b, r, R$ and small $\epsilon > 0$. 
It is inspired by the integrality gap instances by Baier et al.~\cite{BEHKKPSS10} Mahjoub and and McCormick~\cite{MM10}, and made such that the vertex cuts that correspond to {\em dictators} behave the same as the fractional solution that cuts $\frac{1}{r}$ fraction of every vertex. 
All graphs in this section are undirected. 

For positive integers $a, b, r, R$, and $\epsilon > 0$, define $\calD^{\SV}_{a, b, r, R, \epsilon} = (V, E)$ be the graph defined as follows. 
Consider the probability space $(\Omega, \mu)$ where $\Omega := \{ 0, \dots, r-1, * \}$, and $\mu : \Omega \mapsto [0, 1]$ with $\mu(*) = \epsilon$ and $\mu(x) = \frac{1 - \epsilon}{r}$ for $x \neq *$. 

\begin{itemize}
\item $V = \{ s, t \} \cup \{ v^i_x \}_{0 \leq i \leq b, x \in \Omega^R}$. Let $v^i$ denote the set of vertices $\{ v^i_x \}_x$. 

\item For $0 \leq i \leq b$ and $x \in \Omega^R$, $\w(v^i_x) = \mu^{\otimes R}(x)$. Note that the sum of weights is $b + 1$. 

\item For any $0 \leq i \leq b$, there are edges from $s$ to each vertex in $v_i$ with length $ai + 1$ and edges from each vertex in $v_i$ to $t$ with length $(b - i)a + 1$.

\item For $x, y \in \Omega^R$, we call that $x$ and $y$ are {\em compatible} if 
\begin{itemize}
\item For any $1 \leq j \leq R$: [$y_j = (x_j + 1)\mod r$] or [$y_j = *$] or [$x_j = *$]. 
\end{itemize}

\item For any $0 \leq i < b$ and compatible $x, y \in \Omega^R$, we have an edge $(v^i_x, v^{i+1}_y)$ of length $1$ (called a {\em short edge}). 

\item For any $i, j$ such that $0 \leq i < j - 1 < b$ and compatible $x, y \in \Omega^R$, we have an edge $(v^i_x, v^{j}_y)$ of length $(j - i)a$ (called a {\em long edge}).

\end{itemize}

\paragraph{Completeness.} 
We first prove that vertex cuts that correspond to {\em dictators} behave the same as the fractional solution that gives $\frac{1}{r}$ to every vertex. For any $q \in [R]$, let $V_q := \{ v^i_x : 0 \leq i \leq b, x_q = * \mbox{ or } 0 \}$. Note that the total weight of $V_q$ is $(b + 1)(\epsilon + \frac{1 - \epsilon}{r})$. 

\begin{lem}
After removing vertices in $V_q$, the length of the shortest path is at least $a(b - r + 2)$.
\label{lem:dict_vertex}
\end{lem}
\begin{proof}
Let $p = (s, v^{i_1}_{x^1}, \dots, v^{i_z}_{x^z}, t)$ be a path from $s$ to $t$ where $i_j \in \{ 0, \dots, b \}$ and $x^j \in \Omega^R$ for each $1 \leq j \leq z$. Let $y_j := (x^j)_q \in \{ 0, \dots, r - 1 \}$ for each $1 \leq j \leq z$. 

For each $1 \leq j < z$, the edge $(v^{i_j}_{x^j}, v^{i_{j + 1}}_{x^{j +1}})$ is either a long edge or a short edge, and either taken forward (i.e., $i_j < i_{j + 1}$) or backward (i.e., $i_j > i_{j + 1}$). Let $z_{\SLF}, z_{\SSF}, z_{\SLB}, z_{\SSB}$ be the number of long edges taken forward, short edges taken forward, long edges taken  backward, and shot edges taken backward, respectively ($z_{\SLF} + z_{\SSF} + z_{\SLB} + z_{\SSB} = z - 1$). For $1 \leq j \leq z_{\SLF}$ (resp. $z_{\SLB}$), consider the $j$th long edge taken forward (resp. backward) --- it is $(v^{i_{j'}}_{x^{j'}}, v^{i_{j'+1}}_{x^{j'+1}})$ for some $j'$. Let $s^{\SF}_j$ (resp. $s^{\SB}_j$) be $|i_{j'} - i_{j' + 1}|$. The following equality holds by observing how $i_j$ changes. 
\begin{equation}
i_1 + \sum_{j=1}^{z_{\SLF}} s^{\SF}_j + 
z_{\SSF} - \sum_{j=1}^{z_{\SLB}} s^{\SB}_j - z_{\SSB} = i_{z} 
\quad \Rightarrow \quad 
i_1 + \sum_{j=1}^{z_{\SLF}} s^{\SF}_j + 
z_{\SSF} - z_{\SLB} - z_{\SSB} - i_{z} \geq 0.
\label{eq:vertex1}
\end{equation}
Consider how $y_j$ changes. Taking any edge forward increases $y_j$, and taking any edge backward decreases $y_j$. Since $y_j$ can never be $0$ or $*$, we can conclude that 
\begin{equation}
 z_{\SLF} + z_{\SSF} -  z_{\SLB} - z_{\SSB} \leq r - 2.
\label{eq:vertex2}
\end{equation}
$\eqref{eq:vertex1} - \eqref{eq:vertex2}$ yields
\begin{equation}
i_1  - i_{z} + \sum_{j=1}^{z_{\SLF}} (s^{\SF}_j - 1) 
\geq 2 - r 
\quad \Rightarrow \quad 
i_1  - i_{z} + \sum_{j=1}^{z_{\SLF}} s^{\SF}_j 
\geq 2 - r.
\label{eq:vertex3}
\end{equation}
The total length of $p$ is 
\begin{align*}
& 2 + a \big( i_1 + b - i_{z} + \sum_{j=1}^{z_{\SLF}} s^{\SF}_j + \sum_{j=1}^{z_{\SLB}} s^{\SB}_j) + z_{\SSF} + z_{\SSB}  \\
\geq \quad & a \big( i_1 + b - i_{z} + \sum_{j=1}^{z_{\SLF}} s^{\SF}_j) \\
\geq \quad & a(b - r + 2).
\end{align*}
\end{proof}

\paragraph{Soundness.}
To analyze soundness, we define a correlated probability space $(\Omega_1 \times \Omega_2, \nu)$ where both $\Omega_1, \Omega_2$ are copies of $\Omega = \{ 0, \dots, r - 1, * \}$. It is defined by the following process to sample $(x, y) \in \Omega^2$.
\begin{itemize}
\item Sample $x \in \{0, \dots, r - 1 \}$. Let $y = (x + 1) \mod r$. 
\item Change $x$ to $*$ with probability $\epsilon$. Do the same for $y$ independently. 
\end{itemize}
Note that the marginal distribution of both $x$ and $y$ is equal to $\mu$. Assuming $\epsilon < \frac{1}{2r}$, the minimum probability of any atom in $\Omega_1 \times \Omega_2$ is $\epsilon^2$.
Furthermore, in our correlated space, $\nu(x, *) > 0$ for all $x \in \Omega_1$ and $\nu(*, x) > 0$ for all $x \in \Omega_2$. Therefore, we can apply Lemma~\ref{lem:mossel_connected} to conclude that 
$\rho(\Omega_1, \Omega_2 ; \nu) \leq \rho := 1 - \frac{\epsilon^4}{2}$.

Apply Theorem~\ref{thm:mossel} ($\rho \leftarrow \rho, \alpha \leftarrow \epsilon^2, \epsilon \leftarrow \frac{\Gamma_{\rho}(\frac{\epsilon}{3}, \frac{\epsilon}{3})}{2} $) to get $\tau$ and $d$. We will later apply this theorem with the parameters obtained here. 
Fix an arbitrary subset $C \subseteq V$, and $C_i := C \cap v^i$.  
For $0 \leq i \leq b$, call $v^i$ {\em blocked} if $\mu^{\otimes R} [C_i(x)] \geq 1 - \epsilon$. At most $\lfloor \frac{\w(C)}{1 - \epsilon} \rfloor$ $v^i$'s can be blocked. Let $k'$ be the number of blocked $v^i$'s, and $z = b + 1 - k'$ be the number of unblocked $v^i$'s. Let $\{ v^{i_1}, \dots, v^{i_{z}} \}$ be the set of unblocked $v^i$'s with $i_1 < i_2 < \dots < i_{z}$. 

For $1 \leq j \leq z$, let $S_j \subseteq v^{i_j}$ be such that $x \in S_j$ if there exists a path $(p_0 = s, p_1, \dots, p_{j - 1}, v^{i_j}_x)$ such that each $p_{j'} \in v^{i_{j'}} \setminus C$ ($1 \leq j' < j$). 
For $1 \leq j \leq z$, let $f_j : \Omega^R \mapsto [ 0, 1  ]$ be the indicator function of $S_j$.

We prove that if none of $f_j$ reveals any {\em influential coordinate}, $\mu^{\otimes R} (S_z) > 0$. 
Since any path passing $v^{i_1}, \dots, v^{i_z}$ (bypassing only blocked $v^i$'s) uses short edges at least $b - 2k'$ times, so the length of the shortest path after removing $C$ is at most $2 + (b - 2k') + 2 a k'$. 

\begin{lem}
Suppose that for any $1 \leq j \leq z$ and $1 \leq i \leq R$, $\Inf_i^{\leq d}[f_j] \leq \tau$. Then $\mu^{\otimes R}(S_z) > 0$. 
\end{lem}
\begin{proof}
We prove by induction that $\mu^{\otimes R} (S_j) \geq \frac{\epsilon}{3}$. It holds when $j = 1$ since $v^{i_1}$ is unblocked. Assuming $\mu^{\otimes R} (S_j) \geq \frac{\epsilon}{3}$, since $S_j$ does not reveal any influential coordinate, Theorem~\ref{thm:mossel} shows that for any subset $T_{j + 1} \subseteq v^{i_{j + 1}}$ with $\mu^{\otimes R}(T_{j + 1}) \geq \frac{\epsilon}{3}$, there exists an edge between $S_j$ and $T_{j + 1}$. If $S'_{j + 1} \subseteq v^{i_{j + 1}}$ is the set of neighbors of $S_j$, we have $\mu^{\otimes R} (S'_{j + 1}) \geq 1 - \frac{\epsilon}{3}$. Since $v^{i_{j+1}}$ is unblocked, $\mu^{\otimes R} (S'_{j + 1} \setminus C) \geq \frac{2\epsilon}{3}$, completing the induction. 
\end{proof}

In summary, in the completeness case, if we cut vertices of total weight $k := k(a, b, r, \epsilon) = (b + 1)(\epsilon + \frac{1 - \epsilon}{r})$, the length of the shortest path is at least $l := l(a, b, r, \epsilon) = a(b - r + 2)$. 
In the soundness case, even after cutting vertices of total weight $k'$, the length of the shortest path is at most $2 + (b - \frac{k'}{1 - \epsilon}) + 2a (\frac{k'}{1 - \epsilon})$. 

\begin{itemize}
\item Let $a = 4, b = 2r - 2$ and $\epsilon$ small enough so that $k \leq 2, l = 4r$. Requiring $l' \geq l$ results in $k' = \Omega(r)$, giving a gap of $\Omega(r) = \Omega(l)$ for Length Bounded Cut.

\item Let $a = r, b = 2r - 2$ and $\epsilon$ small enough so that $k \leq 2, l = r^2$. Requiring $k' \leq 2$ results in $l' = O(r)$, giving a gap of $\Omega(r) = \Omega(\sqrt{l})$ for Shortest Path Interdiction. Generally, $k' \leq O(r^{\epsilon})$ results in $l' \leq O(r^{1 + \epsilon})$, giving an $(O(r^{\epsilon}), O(r^{1 - \epsilon}))$-bicriteria gap for any $\epsilon \in (0, 1)$. 
\end{itemize}

\section{RMFC}
\label{sec:test_rmfc}
We present our dictatorship test for the RMFC problem. Our test is inspired by the integrality gap example in Chalermsook and Chuzhoy~\cite{CC10}, which is suggested by Khanna and Olver.
This test will be used in Section~\ref{sec:ug} to prove the hardness result based on Conjecture~\ref{conj:ug2}. All graphs in this section are undirected. 
We will prove hardness of RMFC where $T = \{ t \}$ for a single vertex $t$. 

Given positive integers $b$ and $R$, let $B = (b!) \cdot (\sum_{i=1}^b \frac{b!}{i})$, $\Omega = \{ *, 1, \dots, B \}^R$. 
Consider the probability space $(\Omega, \mu)$ where $\mu : \Omega \mapsto [0, 1]$ with $\mu(*) = \epsilon$ and $\mu(x) = \frac{1 - \epsilon}{B}$ for $x \neq *$. 
We define $\calD^{\SF}_{b, R, \epsilon} = (V, E)$ as follows. 

\begin{itemize}
\item $V = \{ s, t \} \cup (\{ v^i_x \}_{1 \leq i \leq b, x \in \Omega^R})$. Let $v^i := \{ v^i_x \}_{x \in \Omega^R}$. The weight a vertex $v^i_x$ is $i \cdot \mu^{\otimes R}(x)$. 

\item There is an edge from $s$ to each vertex in $L_1$, from each vertex in $L_{b}$ to $t$.

\item For $x, y \in \Omega^R$, we call that $x$ and $y$ are {\em compatible} if 
\begin{itemize}
\item For any $1 \leq j \leq R$: [$y_j = x_j$] or [$y_j = *$] or [$x_j = *$]. 
\end{itemize}

\item For any $0 \leq i < b$ and compatible $x, y \in \Omega^R$, we have an edge $(v^i_x, v^{i+1}_y)$. 
\end{itemize}

\paragraph{Completeness.} 
We first prove that vertex cuts that correspond to {\em dictators} are efficient. Let $H_i = 1 + \frac{1}{2} + \dots + \frac{1}{i} = \frac{\sum_{j=1}^i \frac{i!}{j}}{i!}$ be the $i$th harmonic number. 
For $1 \leq i \leq b$, let $B_i = \frac{H_i}{H_b}B$ and $B_0 = 0$. Each $B_i$ is an integer since $B = (b!) \cdot (\sum_{i=1}^b \frac{b!}{i})$, and
$\sum_{i = 1}^b (B_i - B_{i - 1}) = B$. 

For any $q \in [R]$, we consider the solution
where on Day $i$ ($1 \leq i \leq b$), we save 
\[
V_q^i := \{ v^i_x : x_q = * \mbox{ or } B_{i - 1} + 1 \leq x_q \leq B_i \}.
\] 
Note each day the total weight 
that the total weight of $V_q$ is $i(\epsilon + \frac{1}{i \cdot H_b}) \leq b\epsilon + \frac{1}{H_b}$. 

\begin{lem}
In above solution, $t$ is never burnt.
\label{lem:dict_fire}
\end{lem}
\begin{proof}
Fix an arbitrary $p = (s, v^{i_1}_{x^1}, \dots, v^{i_z}_{x^z}, t)$ from $s$ to $t$, and let $y_j = (x^j)_q$ ($1 \leq j \leq z$). 
Since $i_j \leq j$ for any $j$, $V_q^i$ is saved before we arrive $v^{i_j}_{x^j}$. Therefore $y_1 = y_2 = \dots = y_z$. There exists $r \in \{1, \dots, b \}$ such that $y \in \{B_{r - 1} + 1, \dots, B_{r} \}$. $p$ intersects $V^r_q$. 
\end{proof}

\paragraph{Soundness.}
To analyze soundness, we define a correlated probability space $(\Omega_1 \times \Omega_2, \nu)$ where both $\Omega_1, \Omega_2$ are copies of $\Omega = \{ *, 1, \dots, B \}$. It is defined by the following process to sample $(x, y) \in \Omega^2$.
\begin{itemize}
\item Sample $x \in \{1, \dots, B \}$. Let $y = x$. 
\item Change $x$ to $*$ with probability $\epsilon$. Do the same for $y$ independently. 
\end{itemize}
Note that the marginal distribution of both $x$ and $y$ is equal to $\mu$. Assuming $\epsilon < \frac{1}{2B}$, the minimum probability of any atom in $\Omega_1 \times \Omega_2$ is $\epsilon^2$. Furthermore, in our correlated space, $\nu(x, *) > 0$ for all $x \in \Omega_1$ and $\nu(*, x) > 0$ for all $x \in \Omega_2$. Therefore, we can apply Lemma~\ref{lem:mossel_connected} to conclude that 
$\rho(\Omega_1, \Omega_2 ; \nu) \leq \rho := 1 - \frac{\epsilon^4}{2}$.
Apply Theorem~\ref{thm:mossel} ($\rho \leftarrow \rho, \alpha \leftarrow \epsilon^2, \epsilon \leftarrow \frac{\Gamma_{\rho}(\frac{1}{3}, \frac{1}{3})}{2} $) to get $\tau$ and $d$. We will later apply this theorem with the parameters obtained here. 

Fix an arbitrary solution where we save $C_i \subseteq V$ on Day $i$ with $\w(C_i) \leq k'$. Let $S_i \subseteq v^i$ be the set of vertices of $v^i$ burnt at the end of Day $i$. Let $f_i : \Omega^R \mapsto [ 0, 1  ]$ be the indicator function of $S_i$ ($1 \leq i \leq b$). 
We prove that if none of $f_i$ reveals any {\em influential coordinate}, unless $k'$ is large, $\mu^{\otimes R}(S_i)$ is large for all $i$, so $t$ will be burnt on Day $b + 1$.

\begin{lem}
Suppose that for any $1 \leq i \leq b$ and $1 \leq j \leq R$, $\Inf_j^{\leq d}[f_i] \leq \tau$. 
If $k' \leq \frac{1}{3}$, $\mu^{\otimes R}(S_i) \geq \frac{1}{3}$ for all $1 \leq i \leq b$. 
\end{lem}
\begin{proof}
We prove by induction on $i$. 
It is easy to see $\mu^{\otimes R}(S_1) \geq \frac{1}{3}$ since the $\w(v^1) = 1$ but $k' \leq \frac{1}{3}$. Suppose that the claim holds for $i$. For any $T \subseteq v^{i + 1}$ with $\mu^{\otimes R}(T) \leq \frac{1}{3}$, since $S_i$ does not reveal any influential coordinate, Theorem~\ref{thm:mossel} shows that there exists an edge between $S_{i}$ and $T$. It implies that $\mu^{\otimes R}(N(S_i)) \geq \frac{2}{3}$, where $N(S_i) \subseteq v^{i + 1}$ denotes the set of neighbors of $S_i$ in $v^{i + 1}$. The total weight of saved vertices up to Day $i$ is at most $i k' \leq \frac{i}{3}$. 
Since $\w(v^i) = i$, even if all saved vertices are in $v^i$, $\mu^{\otimes R} (v^i \cap (C_1 \cup \dots \cup C_i)) \leq \frac{1}{3}$. Since $S_{i + 1} = N(S_i) \setminus (C_1 \cup \dots \cup C_i)$, $\mu^{\otimes R}(S_{i + 1}) \geq \frac{1}{3}$, the induction is complete. 
\end{proof}

In summary, in the completeness case, 
we save vertices of total weight at most $b\epsilon + \frac{1}{H_b}$ and save $t$. In the soundness case, we fail to save $t$ unless we spend total weight at least $\frac{1}{3}$ each day. By taking $\epsilon$ small enough, the gap becomes $\Omega(\log b)$.

\section{Unique Games Hardness}
\label{sec:ug}

\subsection{UGC and Variant}

We introduce the Unique Games Conjecture and its equivalent variant.

\begin{defi}
An instance $\mathcal{L}(B(U_B \cup W_B, E_B), [R], \left\{ \pi(u, w) \right\}_{(u, w) \in E_B})$ of Unique Games consists of a biregular bipartite graph $B(U_B \cup W_B, E_B)$ and a set $[R]$ of labels. For each edge $(u, w) \in E_B$ there is a constraint specified by a permutation $\pi(u, w) : [R] \rightarrow [R]$. The goal is to find a labeling $l : U_B \cup W_B \rightarrow [R]$ of the vertices such that as many edges as possible are satisfied, where an edge $e = (u, w)$ is said to be satisfied if $l(u) = \pi(u, w)(l(w))$. 
\end{defi}
\begin{defi}
Given a Unique Games instance $\mathcal{L} (B(U_B \cup W_B, E_B), [R], \left\{ \pi(u, w) \right\}_{(u, w) \in E_B})$, let $\mathsf{Opt}(\mathcal{L})$ denote the maximum fraction of simultaneously-satisfied edges of $\mathcal{L}$ by any labeling, i.e.
\begin{equation*}
\mathsf{Opt}(\mathcal{L}) := \frac{1}{|E_B|} \ \max_{l : U_B \cup W_B \rightarrow [R] } |\left\{ e \in E : l \mbox{ satisfies }e \right\} |.
\end{equation*}
\end{defi}
\begin{conj} [The Unique Games Conjecture~\cite{Khot02}] 
\label{conj:ug}
For any constants $\eta > 0$, there is $R = R(\eta)$ such that, for a Unique Games instance $\mathcal{L}$ with label set $[R]$, it is NP-hard to distinguish between
\begin{itemize}
\itemsep=0ex
\item $\mathsf{opt}(\mathcal{L}) \geq 1 - \eta$. 
\item $\mathsf{opt}(\mathcal{L}) \leq \eta$. 
\end{itemize}
\label{conj:ug}
\end{conj}
To show the optimal hardness result for Vertex Cover, Khot and Regev~\cite{KR08} introduced the following seemingly stronger conjecture, and proved that it is in fact equivalent to the original Unique Games Conjecture.

\begin{conj} [Khot and Regev~\cite{KR08}] 
\label{conj:ug_variant}
For any constants $\eta > 0$, there is $R = R(\eta)$ such that, for a Unique Games instance $\mathcal{L}$ with label set $[R]$, it is NP-hard to distinguish between
\begin{itemize}
\itemsep=0ex
\item There is a set $W' \subseteq W_B$ such that $|W'| \geq (1 - \eta)|W_B|$ and a labeling $l : U_B \cup W_B \rightarrow [R]$ that satisfies every edge $(u, w)$ for $v \in U_B$ and $w \in W'$.
\item $\mathsf{opt}(\mathcal{L}) \leq \eta$. 
\end{itemize}
\label{conj:ug1}
\end{conj}

For RMFC, we use the following variant of Unique Games, which is not known to be equivalent to the original Conjecture.
\begin{conj}
\label{conj:ug_variant2}
For any constants $\eta > 0$, there is $R = R(\eta)$ such that, for a Unique Games instance $\mathcal{L}$ with label set $[R]$, it is NP-hard to distinguish between
\begin{itemize}
\itemsep=0ex
\item There is a set $W' \subseteq W_B$ such that $|W'| \geq (1 - \eta)|W_B|$ and a labeling $l : U_B \cup W_B \rightarrow [R]$ that satisfies every edge $(u, w)$ for $v \in U_B$ and $w \in W'$.
\item $\mathsf{opt}(\mathcal{L}) \leq \eta$. Moreover, the instance satisfies the following expansion property: For every set $S \subseteq W_B$, $|S| = \frac{|W_B|}{10}$, we have $|N(S)| \geq \frac{9}{10} |U_B|$, where $N(S) := \{ v \in U_B: \exists w \in S, (v, w) \in E_B \}$. 
\end{itemize}
\label{conj:ug2}
\end{conj}
Conjecture~\ref{conj:ug2} is similar to that of Bansal and Khot~\cite{BK09}, under which the optimal hardness of {\em Minimizing Weighted Completion Time with Precedence Constraints} is proved. Their conjecture requires that in the soundness case, $\forall S \subseteq W_B$ with $|S| = \delta |W_B|$, we must have $|N(S)| \geq (1 - \delta)|U_B|$ for arbitrarily small $\delta$. Our conjecture is a weaker (so more likely to hold) since we require this condition for only one value $\delta = \frac{1}{10}$. 

\subsection{General Reduction}
We now introduce our reduction from Unique Games to our problems \edgecutn, \vertexcutn, Directed Multicut, and RMFC. 
We constructed four dictatorship tests for $\calD^{\SE}_{a, b, r, R}$, $\calD^{\SV}_{a, b, r, R, \epsilon}$, $\calD^{\SM}_{r, k, R, \epsilon}$, and $\calD^{\SF}_{b, R, \epsilon}$. 
Fix a problem, and let $\calD = (\VD, \ED) $ be the dictatorship test for the problem with the chosen parameters. $\calD^{\SE}$ is edge-weighted and $\calD^{\SV}$, $\calD^{\SM}$ and $\calD^{\SF}$ are vertex-weighted, and our reduction will take care of this difference whenever relevant.

Given an instance $\mathcal{L}$ of Unique Games, we describe how to reduce it to a graph $G = (V_G, E_G)$. We assign to each vertex $w \in W_B$ a copy of $\VD$ --- formally, $V_G := \{ s, t \} \cup (W_B \times \VD)$ for \edgecutn, \vertexcutn, RMFC, and 
$V_G := \{ s_i, t_i \}_{i \in [k]} \cup (W_B \times \VD)$ for Directed Multicut. 
For any $w \in W_B, v \in \VD$, the vertex weight of $(w, v)$  is $\frac{\w(v)}{|W_B|}$, so that the sum of vertex weights is $b + 1$ for \vertexcut and $\frac{b(b+1)}{2}$ for RMFC, and $r^k$ for Directed Multicut. 

For a permutation $\sigma : [R] \rightarrow [R]$, let $x \circ \sigma := (x_{\sigma(1)}, \dots, x_{\sigma(R)})$. 
To describe the set of edges, consider the random process where $u \in U_B$ is sampled uniformly at random, and its two neighbors $w^1, w^2$ are independently sampled. 
For each edge $(v^{i_1}_{x^1}, v^{i_2}_{x^2}) \in E_{\calD}$, we create an edge $((w_1, v^{i_1}_{x^1 \circ \pi(u, w^1)}), (w_2, v^{i_2}_{x^2 \circ \pi(u, w^2)}))$. Call this edge is {\em created by} $u$. 
For \edgecut, the weight of each edge is the weight in $\calD^{\SE}$  times the probability that $(u, w_1, w_2)$ are sampled. The sum of weights is $b$. 
For each edge incident on a terminal (i.e., $(X, v^i_x)$ or $(v^i_x, X)$ where $X \in \{ s, t \} \cup \{ s_i, t_i \}_i$), we add the corresponding edge $(X, (w, v^i_x))$ or $((w, v^i_x), X)$ for each $w \in W_B$. For \edgecut, their wegiths are $\infty$ as in $\calD^{\SE}$.

\subsection{Completeness}
Suppose there exists a labeling $l$ and a subset $W' \subseteq W_B$ with $|W'| \geq (1 - \eta)|W_B|$ such that $l$ satisfy every edge incident on $W'$. 

\paragraph{\edgecut.}
For every triple $(u, w_1, w_2)$ such that $u \in U_B$ and $(u, w_1), (u, w_2) \in E_B$, we cut the following edges.
\[
\{ ((w_1, v^i_x), (w_2, v^{i + 1}_y) : 0 \leq i < b, y_{l(w_2)} \neq x_{l(w_1)} + 1 \mod R \mbox{ or } (x_{l(w_1)}, y_{l(w_2)}) = (0, 1) \}. 
\]
For $w \notin W'$, we additionally cut every edge incident on $\{ w \} \times \calD$. 
The total cost is at most $\frac{2b}{r} + 2 \eta b$. 
The completeness analysis for the dictatorship test ensures that the length of the shortest path is at least $a(b - r + 1)$. 
The proof of Lemma~\ref{lem:dict_edge} works if we have $y_j = x^j_{l(w_j)}$. 

\paragraph{\vertexcut.}
For every $w \in W'$, we cut the following vertices.
\[
\{ (w, v^{i}_x) : 0 \leq i \leq b, x_{l(w)} = * \mbox { or } 0 \}. 
\]
For $w \notin W'$, we cut every vertex in $\{ w \} \times \calD$. 
The total cost is $(b + 1)(\epsilon + \frac{1 - \epsilon}{r}) + \eta (b + 1)$. 
The completeness analysis for the dictatorship test ensures that the length of the shortest path is at least $a(b - r + 2)$. 
The proof of Lemma~\ref{lem:dict_vertex} works if we have $y_j = x^j_{l(w_j)}$.

\paragraph{Directed Multicut.}
For every $w \in W'$, we cut the following vertices.
\[
\{ (w, v^{\alpha}_x) : \alpha \in [r]^k, x_{l(w)} = * \mbox { or } 0 \}. 
\]
For $w \notin W'$, we cut every vertex in $\{ w \} \times \calD$. 
The total cost is at most $(\epsilon + \frac{1 - \epsilon}{r})r^k + \eta r^k \leq r^{k-1}(1 + r \epsilon + r \eta)$. 
The completeness analysis for the dictatorship test, Lemma~\ref{lem:dict_vertex}, ensures that there is no path from $s_i$ to $t_i$ for any $i$.

\paragraph{RMFC.}
For $w \in W'$, on Day $i (1 \leq i \leq b)$, we save every vertex in 
\[
\{ (w, v^i_x) : x_{l(w)} = * \mbox { or } B_{i-1} \leq x_{l(w)} \leq B_i \},
\]
where $B_i = \frac{H_i}{H_b}B$. 
For $w \notin W'$, on Day $i \, (1 \leq i \leq b)$, we save every vertex in $(w, v^i)$. This ensures that fire never spreads to vertices associated with $w \notin W'$. 
Each day, the total cost of saved vertices is at most $b\epsilon + \frac{1}{H_b} + b \eta$. 
The completeness analysis for the dictatorship test ensures that $t$ is saved in this case. 
The proof of Lemma~\ref{lem:dict_fire} works if we have $y_j = x^j_{l(w_j)}$.

\subsection{Soundness for Cut / Interdiction Problems}
We present the soundness analysis for \edgecutn, \vertexcutn, and Directed Multicut. The soundness analysis of RMFC is in Section~\ref{subsec:rmfc_sound}.
We first discuss how to extract an influential coordinate for each $u \in U_B$. 

\paragraph{\edgecut.}
Fix an arbitrary $C \subseteq E_G$ with the total cost $k'$, and consider the graph after cutting edges in $C$. 
We will show that if the length of the shortest path is greater than $l' = 2 + b - 4k'\sqrt{r} + 4ak' \sqrt{r}$, we can decode influential coordinates for many vertices of the Unique Games instance. 

For each $w \in W_B$, $0 \leq j \leq b$, and a sequence $\barc = (c_1, \dots, c_j) \in \{ \SL, \SSS \}^j$, let $g_{w, j, \barc} : \Omega^R \mapsto \{ 0, 1 \}$ such that $g_{w, j, \barc}(x) = 1$ if and only if there exists a path $p = (s, p_0 = (w_0, v^0_{x^0}), \dots, p_{j - 1} = (w_{j - 1}, v^{j - 1}_{x^{j - 1}}) , p_j = (w, v^j_x))$ for some $w_0, \dots, w_{j - 1} \in W_B$ and $x^0, \dots, x^{j - 1} \in \Omega^R$ such that $(p_{j' - 1}, p_{j'})$ is long if and only if $c_{j'} = \SL$ for $1 \leq j' \leq j$.

For $u \in U_B, 0 \leq j \leq b$, and $\barc \in \{ \SL, \SSS \}^j$, let $f_{u, j, \barc} : \Omega_R \mapsto [0, 1]$ be such that
\[
f_{u, j, \barc} (x) = \E_{w \in N(u)} [g_{w, j, \barc} (x \circ \pi^{-1} (u, w)) ],
\]
where $N(u)$ is the set of neighbors of $u$ in the Unique Games instance. 

Let $\gamma(u)$ be the sum of weights of the edges created by $u$ in $C$. $\E_u [\gamma(u)] = k'$, so at least $\frac{1}{2}$ fraction of $u$'s have $\E_u [\gamma(u)] \leq 2k'$. For such $u$, since the length of the shortest path is greater than $l' = 2 + b - 4k'\sqrt{r} + 4ak' \sqrt{r}$, the soundness analysis for the dictatorship test shows that there exist $j \in \{0, \dots, b \}, q \in [R], \barc$ such that $\Inf_{q}^{\leq d} [f_{u, j, \barc}] \geq \tau$ ($d$ and $\tau$ do not depend on $u$). 

\paragraph{\vertexcut.}
Fix an arbitrary $C \subseteq V_G$ with the total cost $k'$, and consider the graph after cutting vertices in $C$. 
We will show that if the length of the shortest path is greater than $l' = 2 + (b - 4k') + 8 a k'$, we can decode influential coordinates for many vertices of the Unique Games instance. 

For each $w \in W_B$, $1 \leq j \leq b$, and a sequence $\bari = (i_1 < \dots < i_j) \in \{0, \dots, b \}^j$, let $g_{w, j, \bari} : \Omega^R \mapsto \{ 0, 1 \}$ such that $g_{w, j, \bari}(x) = 1$ if and only if there exists a path $p = (s, (w_1, v^{i_1}_{x^1}), \dots, (w_{j - 1}, v^{i_{j-1}}_{x^{j-1}}), (w, v^{i_j}_x))$ for some $w_1, \dots, w_{j - 1} \in W_B$ and $x^1, \dots, x^{j-1} \in \Omega^R$. 

For $u \in U_B, 0 \leq j \leq b$, and $\bari \in \{ 0, \dots, b \}^i$, let  $f_{u, j, \bari} : \Omega_R \mapsto [0, 1]$ be such that
\[
f_{u, j, \bari} (x) = \E_{w \in N(u)} [g_{w, j, \bari} (x \circ \pi^{-1} (u, w)) ],
\]
where $N(u)$ is the set of neighbors of $u$ in the Unique Games instance. 

Let $\gamma(u)$ be the expected weight of $C \cap (\{ w \} \times \calD)$, where $w$ is a random neighbor of $u$. 
$\E_u [\gamma(u)] = k'$, so at least $\frac{1}{2}$ fraction of $u$'s have $\E_u [\gamma(u)] \leq 2k'$. For such $u$, 
Since the length of the shortest path is greater than $l' = 2 + (b - 4k') + 8 a k'$, the soundness analysis for the dictatorship test shows that there exists $q \in [R], 1 \leq j \leq b, \bari$ such that $\Inf_{q}^{\leq d} [f_{u, j, \bari}] \geq \tau$ ($d$ and $\tau$ do not depend on $u$).

\paragraph{Directed Multicut.}
Fix an arbitrary $C \subseteq V_G$ with the total cost $k'$, and consider the graph after cutting vertices in $C$. 
Let $\beta > 0$ be another small parameter to be determined later. 
If $k' \leq  k(1 - \epsilon)(1 - \beta)(r - 1)^{k - 1}$, we prove that we can decode influential coordinates for many vertices of the Unique Games Instance.

For each $w \in W_B$, $i \in [k]$, $1 \leq j \leq r^k$, and a sequence $\bara = (\alpha_1, \dots, \alpha_j) \in ([r]^k)^j$, let $g_{w, j, \bara} : \Omega^R \mapsto \{ 0, 1 \}$ such that $g_{w, j, \bara}(x) = 1$ if and only if there exists a path $p = (s, (w_1, v^{\alpha_1}_{x^1}), \dots, (w_{j - 1}, v^{\alpha_{j-1}}_{x^{j-1}}), (w, v^{\alpha_j}_x))$ for some $w_1, \dots, w_{j - 1} \in W_B$ and $x^1, \dots, x^{j - 1} \in \Omega^R$.

For $u \in U_B, 0 \leq j \leq b$, and $\bara \in ([r]^k)^j$, let  $f_{u, j, \bara} : \Omega_R \mapsto [0, 1]$ be such that
\[
f_{u, j, \bara} (x) = \E_{w \in N(u)} [g_{w, j, \bara} (x \circ \pi^{-1} (u, w)) ],
\]
where $N(u)$ is the set of neighbors of $u$ in the Unique Games instance. 

Let $\gamma(u)$ be the expected weight of $C \cap (\{ w \} \times \calD)$, where $w$ is a random neighbor of $u$. 
$\E_u [\gamma(u)] = k' \leq  k(1 - \epsilon)(1 - \beta)(r - 1)^{k - 1}$, so at least $\beta$ fraction of $u$'s have $\E_u [\gamma(u)] \leq k(1 - \epsilon)(r - 1)^{k - 1}$. 
For such $u$, since any $s_i$-$t_i$ pair is disconnected, the soundness analysis for the dictatorship test shows that there exists $q \in [R], 1 \leq j \leq r^k, \bara$ such that $\Inf_{j'}^{\leq d} [f_{u, j, \bara}] \geq \tau$ ($d$ and $\tau$ do not depend on $u$).

\paragraph{Finishing Up.} 
The above analyses for \edgecutn, \vertexcutn, and Directed Multicut can be abstracted as follows. Each vertex $u \in U_B$ is associated with $\{ f_{u, h} : \Omega^R \mapsto [0, 1] \}_{h \in T}$ for some index set $H$ ($|H|$ is upper bounded by some function of $b$ for \edgecut and \vertexcutn, and some function of $r$ and $k$ for Multicut).  For at least $\beta$ fraction of $u \in U_B$ ($\beta = \frac{1}{2}$ for \edgecut and \vertexcut), there exist $h \in H$ and $q \in [R]$ such that $\Inf_{q}^{\leq d} [f_{u, h}] \geq \tau$. Set $l(u) = q$ for those vertices. 
Since
\begin{align*}
\Inf_{q}^{\leq d}(f_{u, h}) 
&= \sum_{\alpha_q \neq 0, |\alpha| \leq d} \widehat{f_{u, h}}(\alpha)^2 
= \sum_{\alpha_q \neq 0, |\alpha| \leq d} (\E_w [\widehat{f_{w, h}}(\pi(u, w)^{-1}(\alpha))]^2) \\
&\leq \sum_{\alpha_q \neq 0, |\alpha| \leq d} \E_w [\widehat{f_{w, h}}(\pi(u, w)^{-1}(\alpha))^2] = \E_w [\Inf_{\pi(u, w)^{-1}(q)}^{\leq d}(f_{w, h})],
\end{align*}
at least $\tau/2$ fraction of $u$'s neighbors satisfy $\Inf_{\pi(u, w)^{-1}(q)}^{\leq d} (f_{w, h}) \geq \tau / 2$. There are at most $2d / \tau$ coordinates with degree-$d$ influence at least $\tau / 2$ for a fixed $h$, so their union over $h \in H$ yields at most $\frac{2d \cdot |H|}{\tau}$ coordinates. Choose $l(w)$ uniformly at random among those coordinates (if there is none, set it arbitrarily). The above probabilistic strategy satisfies at least $\beta(\frac{\tau}{2})(\frac{\tau}{2d \cdot |H|} )$ fraction of all edges. Taking $\eta$ smaller than this quantity proves the soundness of the reductions. 

\subsection{Soundness for RMFC}
\label{subsec:rmfc_sound}
Fix an arbitrary solution $C_1, \dots, C_b \subseteq V$ such that $C_i$ is saved on Day $i$ and the weight of each $C_i$ is at most $k' = \frac{1}{10}$. 
Suppose that $t$ is saved. 
We will prove that the Unique Games instance admits a good labeling. 

For each $w \in W_B$, $1 \leq i \leq b$, let $g_{w, i} : \Omega^R \mapsto \{ 0, 1 \}$ such that $g_{w, i}(x) = 1$ if and only if $(w, v^i_x)$ is burning on Day $i$. Let Day $i^*$ be the first day where $\E_{w, x}[g_{w, i^*}(x)] \geq \frac{1}{2}$ and  $\E_{w, x}[g_{w, i^*+1}(x)] \leq \frac{1}{2}$. Such $i^*$ must exist since $\E_{w, x}[g_{w, 1}] \geq 1 - k' \geq \frac{1}{2}$ but $\E_{w, x}[g_{w, b}] = 0$. 
For each $w \in W_B$, let $g_w := g_{w, i^*}$ and let $f_w : \Omega^R \mapsto \{0, 1\}$ be such that $f_w(x) = 1$ if and only if there exists $(w', x')$ such that the vertex $(w', v^{i^*}_{x'})$ is burning on Day $i$ and there exists an edge $((w', v^{i^*}_{x'}), (w, v^{i^*+1}_x))$. 
We must have $\E_{w, x}[f_w(x)] \leq \frac{1}{2} + \frac{1}{10} = \frac{3}{5}$, since we can save at most $k' = \frac{1}{10}$ fraction of $\{ g_{w, i^* + 1} \}_w$ before Day $i^* + 1$.

By an averaging argument, at least $\frac{1}{4}$ fraction of $w \in W_B$ satisfies $\E_x [g_{w, i^*}] \geq \frac{1}{4}$. Call them {\em heavy} vertices. By the expansion of the Unique Games instance, at least $\frac{9}{10}$ fraction of $u \in U_B$ has a heavy neighbor,
and at least $\frac{9}{10}$ fraction of $w \in W_B$ has a heavy $w' \in W_B$ such that $(u, w), (u, w') \in E_B$ for some $u \in U_B$ (say $w$ is {\em reachable} from $w'$). 

By Theorem~\ref{thm:mossel}, there exist $\tau$ and $d$ such that for each heavy $w'$, if $\Inf_j^{\leq d} [g_{w'}] \leq \tau$ for all $j \in [R]$, all $w$ reachable from $w'$ should satisfy $\E_x [f_w(x)] \geq \frac{9}{10}$ (say $w'$ {\em reveals an influential coordinate} if such $j$ exists). 
At least $\frac{1}{4} - \frac{1}{10} = 0.15$ fraction of $w'$ are heavy and reveal an influential coordinate, since otherwise by the expansion $\E_{w, x} [f_w(x)] \geq (\frac{9}{10})^2 > \frac{3}{5}$. 

Another expansion argument ensures that at least $\frac{9}{10}$ fraction of $u \in U_B$ is a neighbor of heavy $w$ with an influential coordinate. Call such $u$ {\em good} and let $h_u : \Omega^R \mapsto \{0, 1 \}$ such that $h_u(x) = g_w(x \circ \pi^{-1}(u, w))$. Finally, call $w \in W_B$ {\em good} if $\E_x[f_w(x)] \leq \frac{9}{10}$. Since 
$\E_{w, x}[f_w(x)] \leq \frac{3}{5}$, the fraction of good $w$ is at least $\frac{1}{3}$. Theorem~\ref{thm:mossel} ensures that if there is $(u, w) \in E_B$ where both $u$ and $w$ are good, there exists $j \in [R]$ such that $\min(\Inf_j^{\leq d}[h_u], \Inf_{\pi(u, w)^{-1}(j)}^{\leq d}[f_w]) \geq \tau$. 

Our labeling strategy for Unique Games is as follows. Each good $u$ will get a random label from $\{ j : \Inf_{j}^{\leq d} [h_u] \geq \tau \}$, and each good $w$ will get a random label from $\{ j : \Inf_{j}^{\leq d} [f_w] \geq \tau \}$. Other vertices get an arbitrary label. 
Since at least $\frac{9}{10}$ fraction of $u \in U_B$ are good, $\frac{1}{3}$ fraction of $w \in W_B$ are good, and the Unique Games instance is biregular, at least $\frac{9}{10} - \frac{2}{3} \geq \frac{1}{5}$ fraction of edges are between good vertices. For each $f_w$ or $h_u$, the number of coordinates $j$ with degree-$d$ influence at least $\tau$ is at most $\frac{d}{\tau}$. Therefore, this strategy satisfies at least 
$\frac{1}{5} \cdot (\frac{d}{\tau})^2$ fraction of edges in expectation. 
Taking $\eta$ smaller than this quantity proves the soundness of the reduction. 

\subsection{Final Results}
Combining our completeness and soundness analyses and taking $\epsilon$ and $\eta$ small enough, we prove our main results.

\paragraph{\edgecut.} It is hard to distinguish the following cases. 
\begin{itemize}
\item Completeness: There is a cut of weight at most $k := \frac{2b}{r} + 2 \eta b$ such that the length of the shortest path after the cut is at least $l := a(b - r + 1)$. 
\item Soundness: For every cut of weight $k'$, the length of the shortest path is at most $l' := 2 + b - 4k' \sqrt{r} + 4ak'\sqrt{r}$. 
\end{itemize}
Setting $a = 4$, $b = 2r - 1$ yields $k \leq 4$ and $l = 4r$. Since $l' \geq 4r$ implies $k' = \Omega(\sqrt{r})$, we prove the first case of Theorem~\ref{thm:edge1}. Setting $a = \sqrt{r}$ and $b = 2r - 1$ yields $k \leq 4$ and $l = r^{1.5}$. Since $l' = O(k' r)$, we prove the last two cases of Theorem~\ref{thm:edge1}. 

\paragraph{\vertexcut.} It is hard to distinguish the following cases. 
\begin{itemize}
\item Completeness: There is a cut of weight at most $k := (b + 1)(\epsilon + \frac{1 - \epsilon}{r}) + \eta (b + 1)$ such that the length of the shortest path after the cut is at least $l := a(b - r + 2)$. 
\item Soundness: For every cut of weight $k'$, the length of the shortest path is at most $l' := 2 + (b - 4k') + 8 a k'$. 
\end{itemize}
Setting $a = 4$, $b = 2r - 2$ yields $k \leq 2$ and $l = 4r$. Since $l' \geq 4r$ implies $k' = \Omega(r)$, we prove the first case of Theorem~\ref{thm:vertex1}. Setting $a = r$ and $b = 2r - 2$ yields $k \leq 2$ and $l = r^2$. Since $l' = O(k' r)$, we prove the last two cases of Theorem~\ref{thm:vertex1}.

\paragraph{Directed Multicut.} It is hard to distinguish the following cases. 
\begin{itemize}
\item Completeness: There is a cut of weight at most $r^{k-1}(1 + r \epsilon + r \eta)$ that separates every $s_i$ and $t_i$. 
\item Soundness: Every multicut must have weight at least $k(1 - \epsilon)(1 - \beta)(r - 1)^{k - 1}$. 
\end{itemize}
This immediately implies Theorem~\ref{thm:multicut} by taking large $r$ and small $\epsilon, \beta, \eta$.

\paragraph{RMFC.} It is hard to distinguish the following cases. 
\begin{itemize}
\item Completeness: There is a solution where we save vertices of cost $b \epsilon + \frac{1}{H_b} + b\eta = O(\frac{1}{\log b})$ each day to eventually save $t$. 
\item Soundness: Saving vertices of $\frac{1}{10}$ each day cannot save $t$. 
\end{itemize}
This immediately implies Theorem~\ref{thm:firefighter} by taking small $\epsilon$ and $\eta$. 

\paragraph{Acknowledgments}
The author thanks Konstantin Makarychev for useful discussions on Directed Multicut, and Marek Elias for introducing Shortest Path Interdiction. 

\bibliographystyle{alpha}
\bibliography{../../../mybib}

\newcommand{\etalchar}[1]{$^{#1}$}
\begin{thebibliography}{KKMO07}

\bibitem[AAC07]{AAC07}
Amit Agarwal, Noga Alon, and Moses~S. Charikar.
\newblock Improved approximation for directed cut problems.
\newblock In {\em Proceedings of the Thirty-ninth Annual ACM Symposium on
  Theory of Computing}, STOC '07, pages 671--680, New York, NY, USA, 2007. ACM.

\bibitem[ABZ16]{ABZ16}
David Adjiashvili, Andrea Baggio, and Rico Zenklusen.
\newblock Firefighting on trees beyond integrality gaps.
\newblock {\em arXiv preprint arXiv:1601.00271}, 2016.

\bibitem[ACHS12]{ACHS12}
Elliot Anshelevich, Deeparnab Chakrabarty, Ameya Hate, and Chaitanya Swamy.
\newblock Approximability of the firefighter problem.
\newblock {\em Algorithmica}, 62(1-2):520--536, 2012.

\bibitem[BB08]{BB08}
Halil Bayrak and Matthew~D Bailey.
\newblock Shortest path network interdiction with asymmetric information.
\newblock {\em Networks}, 52(3):133--140, 2008.

\bibitem[BEH{\etalchar{+}}10]{BEHKKPSS10}
Georg Baier, Thomas Erlebach, Alexander Hall, Ekkehard K\"{o}hler, Petr Kolman,
  Ond\v{r}ej Pangr\'{a}c, Heiko Schilling, and Martin Skutella.
\newblock Length-bounded cuts and flows.
\newblock {\em ACM Transactions on Algorithms}, 7(1):4:1--4:27, December 2010.
\newblock Preliminary version in ICALP '06.

\bibitem[BGV89]{BGV89}
Michael~O Ball, Bruce~L Golden, and Rakesh~V Vohra.
\newblock Finding the most vital arcs in a network.
\newblock {\em Operations Research Letters}, 8(2):73--76, 1989.

\bibitem[BK09]{BK09}
N.~Bansal and S.~Khot.
\newblock Optimal long code test with one free bit.
\newblock In {\em Proceedings of the 50th Annual IEEE Symposium on Foundations
  of Computer Science}, FOCS '09, pages 453--462, Oct 2009.

\bibitem[BNN15]{BNN15}
Cristina Bazgan, André Nichterlein, and Rolf Niedermeier.
\newblock A refined complexity analysis of finding the most vital edges for
  undirected shortest paths.
\newblock In Vangelis~Th. Paschos and Peter Widmayer, editors, {\em Algorithms
  and Complexity}, volume 9079 of {\em Lecture Notes in Computer Science},
  pages 47--60. Springer International Publishing, 2015.

\bibitem[CC10]{CC10}
Parinya Chalermsook and Julia Chuzhoy.
\newblock Resource minimization for fire containment.
\newblock In {\em Proceedings of the Twenty-first Annual ACM-SIAM Symposium on
  Discrete Algorithms}, SODA '10, pages 1334--1349, Philadelphia, PA, USA,
  2010. Society for Industrial and Applied Mathematics.

\bibitem[CD82]{CD82}
HW~Corley and Y~Sha David.
\newblock Most vital links and nodes in weighted networks.
\newblock {\em Operations Research Letters}, 1(4):157--160, 1982.

\bibitem[CK09]{CK09}
Julia Chuzhoy and Sanjeev Khanna.
\newblock Polynomial flow-cut gaps and hardness of directed cut problems.
\newblock {\em Journal of the ACM}, 56(2), 2009.

\bibitem[CM16]{CM16}
Chandra Chekuri and Vivek Madan.
\newblock Simple and fast rounding algorithms for directed and node-weighted
  multiway cut.
\newblock In {\em Proceedings of the Twenty-Seventh Annual ACM-SIAM Symposium
  on Discrete Algorithms}, pages 797--807, 2016.

\bibitem[CV16]{CV16}
Parinya Chalermsook and Daniel Vaz.
\newblock New integrality gap results for the firefighters problem on trees.
\newblock {\em arXiv preprint arXiv:1601.02388}, 2016.

\bibitem[DK15]{DK15}
Pavel Dvo\v{r}\'{a}k and Du\v{s}an Knop.
\newblock Parametrized complexity of length-bounded cuts and multi-cuts.
\newblock In Rahul Jain, Sanjay Jain, and Frank Stephan, editors, {\em Theory
  and Applications of Models of Computation}, volume 9076 of {\em Lecture Notes
  in Computer Science}, pages 441--452. Springer International Publishing,
  2015.

\bibitem[EVW13]{EVW13}
Alina Ene, Jan Vondr{\'a}k, and Yi~Wu.
\newblock Local distribution and the symmetry gap: Approximability of multiway
  partitioning problems.
\newblock In {\em Proceedings of the Twenty-Fourth Annual ACM-SIAM Symposium on
  Discrete Algorithms}, pages 306--325. SIAM, 2013.

\bibitem[FHNN15]{FHNN15}
Till Fluschnik, Danny Hermelin, Andr{\'e} Nichterlein, and Rolf Niedermeier.
\newblock Fractals for kernelization lower bounds, with an application to
  length-bounded cut problems.
\newblock {\em arXiv preprint arXiv:1512.00333}, 2015.

\bibitem[GL16]{GL16}
Venkatesan Guruswami and Euiwoong Lee.
\newblock Simple proof of hardness of feedback vertex set.
\newblock {\em Theory of Computing}, 2016.
\newblock To appear (as a note).

\bibitem[GSS15]{GSS15}
Venkatesan Guruswami, Sushant Sachdeva, and Rishi Saket.
\newblock Inapproximability of minimum vertex cover on k-uniform k-partite
  hypergraphs.
\newblock {\em SIAM Journal on Discrete Mathematics}, 29(1):36--58, 2015.

\bibitem[GT11]{GT11}
Petr~A. Golovach and Dimitrios~M. Thilikos.
\newblock Paths of bounded length and their cuts: Parameterized complexity and
  algorithms.
\newblock {\em Discrete Optimization}, 8(1):72 -- 86, 2011.
\newblock Parameterized Complexity of Discrete Optimization.

\bibitem[GVY94]{GVY94}
Naveen Garg, Vijay~V Vazirani, and Mihalis Yannakakis.
\newblock Multiway cuts in directed and node weighted graphs.
\newblock In {\em Automata, Languages and Programming}, pages 487--498.
  Springer, 1994.

\bibitem[Har95]{Hartnell95}
Bert Hartnell.
\newblock Firefighter! an application of domination. presentation.
\newblock In {\em 25th Manitoba Conference on Combinatorial Mathematics and
  Computing, University of Manitoba in Winnipeg, Canada}, 1995.

\bibitem[IW02]{IW02}
Eitan Israeli and R~Kevin Wood.
\newblock Shortest-path network interdiction.
\newblock {\em Networks}, 40(2):97--111, 2002.

\bibitem[KBB{\etalchar{+}}07]{KBBEGRZ07}
Leonid Khachiyan, Endre Boros, Konrad Borys, Khaled Elbassioni, Vladimir
  Gurvich, Gabor Rudolf, and Jihui Zhao.
\newblock On short paths interdiction problems: Total and node-wise limited
  interdiction.
\newblock {\em Theory of Computing Systems}, 43(2):204--233, 2007.

\bibitem[Kho02]{Khot02}
Subhash Khot.
\newblock On the power of unique 2-prover 1-round games.
\newblock In {\em Proceedings of the 34th annual ACM Symposium on Theory of
  Computing}, STOC '02, pages 767--775, 2002.

\bibitem[KKMO07]{KKMO07}
Subhash Khot, Guy Kindler, Elchanan Mossel, and Ryan O'Donnell.
\newblock Optimal inapproximability results for {M}ax-{C}ut and other
  2-variable {CSP}s?
\newblock {\em SIAM Journal on Computing}, 37(1):319--357, 2007.

\bibitem[KM10]{KM10}
Andrew King and Gary MacGillivray.
\newblock The firefighter problem for cubic graphs.
\newblock {\em Discrete Mathematics}, 310(3):614--621, 2010.

\bibitem[KMTV11]{KMTV11}
Amit Kumar, Rajsekar Manokaran, Madhur Tulsiani, and Nisheeth~K. Vishnoi.
\newblock On {LP}-based approximability for strict {CSP}s.
\newblock In {\em Proceedings of the Twenty-Second Annual ACM-SIAM Symposium on
  Discrete Algorithms}, SODA '11, pages 1560--1573. SIAM, 2011.

\bibitem[KR08]{KR08}
Subhash Khot and Oded Regev.
\newblock Vertex cover might be hard to approximate to within $2-\epsilon$.
\newblock {\em Journal of Computer and System Sciences}, 74(3):335--349, 2008.

\bibitem[LNLP78]{LNP78}
L{\'a}szl{\'o} Lov{\'a}sz, V~Neumann-Lara, and M~Plummer.
\newblock Mengerian theorems for paths of bounded length.
\newblock {\em Periodica Mathematica Hungarica}, 9(4):269--276, 1978.

\bibitem[MM10]{MM10}
A~Ridha Mahjoub and S~Thomas McCormick.
\newblock Max flow and min cut with bounded-length paths: complexity,
  algorithms, and approximation.
\newblock {\em Mathematical programming}, 124(1-2):271--284, 2010.

\bibitem[MMG89]{MMG89}
K.~Malik, A.~K. Mittal, and S.~K. Gupta.
\newblock The k most vital arcs in the shortest path problem.
\newblock {\em Oper. Res. Lett.}, 8(4):223--227, August 1989.

\bibitem[MNRS08]{MNRS08}
Rajsekar Manokaran, Joseph~Seffi Naor, Prasad Raghavendra, and Roy Schwartz.
\newblock Sdp gaps and ugc hardness for multiway cut, 0-extension, and metric
  labeling.
\newblock In {\em Proceedings of the fortieth annual ACM symposium on Theory of
  computing}, pages 11--20. ACM, 2008.

\bibitem[Mor11]{Morton11}
David~P Morton.
\newblock Stochastic network interdiction.
\newblock {\em Wiley Encyclopedia of Operations Research and Management
  Science}, 2011.

\bibitem[Mos10]{Mossel10}
Elchanan Mossel.
\newblock Gaussian bounds for noise correlation of functions.
\newblock {\em Geometric and Functional Analysis}, 19(6):1713--1756, 2010.

\bibitem[NZ01]{NZ01}
Joseph Naor and Leonid Zosin.
\newblock A 2-approximation algorithm for the directed multiway cut problem.
\newblock {\em SIAM Journal on Computing}, 31(2):477--482, 2001.

\bibitem[OFN12]{OFN12}
Kazumasa Okumoto, Takuro Fukunaga, and Hiroshi Nagamochi.
\newblock Divide-and-conquer algorithms for partitioning hypergraphs and
  submodular systems.
\newblock {\em Algorithmica}, 62(3-4):787--806, 2012.

\bibitem[PS13]{PS13}
Feng Pan and Aaron Schild.
\newblock Interdiction problems on planar graphs.
\newblock In Prasad Raghavendra, Sofya Raskhodnikova, Klaus Jansen, and
  JoséD.P. Rolim, editors, {\em Approximation, Randomization, and
  Combinatorial Optimization. Algorithms and Techniques}, volume 8096 of {\em
  Lecture Notes in Computer Science}, pages 317--331. Springer Berlin
  Heidelberg, 2013.

\bibitem[Rag08]{Raghavendra08}
Prasad Raghavendra.
\newblock Optimal algorithms and inapproximability results for every {CSP}?
\newblock In {\em Proceedings of the 40th annual ACM symposium on Theory of
  computing}, STOC '08, pages 245--254, 2008.

\bibitem[SPG13]{SPG13}
J~Cole Smith, Mike Prince, and Joseph Geunes.
\newblock Modern network interdiction problems and algorithms.
\newblock In {\em Handbook of Combinatorial Optimization}, pages 1949--1987.
  Springer, 2013.

\bibitem[SSZ04]{SSZ04}
Michael Saks, Alex Samorodnitsky, and Leonid Zosin.
\newblock A lower bound on the integrality gap for minimum multicut in directed
  networks.
\newblock {\em Combinatorica}, 24(3):525--530, 2004.

\bibitem[Sve13]{Svensson13}
Ola Svensson.
\newblock Hardness of vertex deletion and project scheduling.
\newblock {\em Theory of Computing}, 9(24):759--781, 2013.
\newblock Preliminary version in APPROX '12.

\bibitem[Wen13]{Wenner13}
Cenny Wenner.
\newblock Circumventing $d$-to-$1$ for approximation resistance of satisfiable
  predicates strictly containing parity of width four.
\newblock {\em Theory of Computing}, 9(23):703--757, 2013.

\end{thebibliography}

\end{document}